%% file: main.tex
\documentclass[sigconf]{acmart}

\usepackage[english]{babel}
\usepackage{blindtext}
\usepackage{multirow}

\usepackage[ruled,vlined,linesnumbered]{algorithm2e}
\usepackage{subfigure}
\usepackage[noend]{algpseudocode}

\renewcommand\footnotetextcopyrightpermission[1]{} 
\setcopyright{none}

\acmDOI{}

\acmISBN{}

\acmPrice{}

\begin{document}
\title{Efficient Routing Algorithm Design for Large DetNet}


\author{Shizhen Zhao, Tianyu Zhu, Ximeng Liu}

\renewcommand{\shortauthors}{X.et al.}

\input{sections/0abstract}

\maketitle

\input{sections/1introduction}

\input{sections/2background}

\input{sections/3model}
\input{sections/4algorithm}
\input{sections/5disjoint}
\input{sections/6evaluation}

\input{sections/7conclusion}

\clearpage
\bibliographystyle{ACM-Reference-Format}
\bibliography{ref}

\appendix
\input{sections/appendix1}
\input{sections/appendix2}
\input{sections/appendix3}

\end{document}

%% file: sections/0abstract.tex
\begin{abstract}
Deterministic Networking (DetNet) is a rising technology that offers deterministic delay \& jitter and zero packet loss regardless of failures in large IP networks. In order to support DetNet, we must be able to find a set of low-cost routing paths for a given node pair subject to delay-range constraints. Unfortunately, the \textbf{Delay-Range} Constrained Routing (DRCR) problem is NP-Complete. Existing routing approaches either cannot support the delay-range constraints, or incur extremely high computational complexity.



We propose Pulse$+$, a highly scalable and efficient DRCR problem solver. Pulse$+$ adopts a branch-and-bound methodology and optimizes its pruning strategies for higher efficiency. We also integrate Pulse$+$ with a divide-and-conquer approach and propose CoSE-Pulse$+$ to find a pair of active/backup paths that meet DetNet's delay-range and delay-diff constraints. Both Pulse$+$ and CoSE-Pulse$+$ offer optimality guarantee. Notably, although Pulse$+$ and CoSE-Pulse$+$ do not have a polynomial worst-case time complexity, their empirical performance is superior. We evaluate Pulse$+$ and CoSE-Pulse$+$ against the K-Shorst-Path and Lagrangian-dual based algorithms using synthetic test cases generated over networks with thousands of nodes and links. Both Pulse$+$ and CoSE-Pulse$+$ achieve significant speedup. To enable reproduction, we open source our code and test cases at \cite{repository}.
\end{abstract}


%% file: sections/1introduction.tex
\section{Introduction}

Modern mission-critical real-time network applications, e.g., tele-surgery~\cite{zhang2022application}, Eastern-Data-Western-Computing~\cite{li2022eastern}, etc., require bounded end-to-end delay \& jitter and zero packet loss rate even under extreme scenarios with network failures. However, existing internet is designed on a best-effort basis, and thus may not be able to support such applications with stringent Quality-of-Service (QoS) requirements. To deal with the above challenges, DetNet Architecture~\cite{rfc8655} was proposed to offer bounded delay and bounded delay jitter guarantee. DetNet aims to achieve bounded delay through Cycle Specified Queuing and Forwarding (CSQF), and guards against network failures using 1+1 path protection. However, it remains an open problem to find a pair of paths that meet DetNet's stringent delay requirements.

We model DetNet's routing problem as an Srlg-disjoint Delay-Range Constrained Routing (DRCR) problem. Shared risk link group (Srlg) is a widely adopted concept to guard against network failures. An Srlg contains a set of links that share a common physical resource (cable, conduit, node, etc.). An Srlg is typically considered as an independent failure domain, and a failure of an Srlg will cause all links in this Srlg to fail simultaneously. To achieve service protection in DetNet, we need to find two paths that do not share any Srlg. With a pair of active/backup paths, DetNet packets are replicated at the source and then transmitted along both paths and finally de-duplicated at the destination. To ensure deterministic delay under path failures, the end-to-end delay of both paths cannot differ too much. This introduces a delay-range constraint to DetNet's routing problem. Unfortunately, existing routing algorithms cannot support the delay-range constraint or scale to large networks with prohibitively high complexity simultaneously.

We solve DetNet's routing problem in two steps. First, given an active path, we solve the Delay-Range Constrained Routing (DRCR) problem to find a single backup path that meets DetNet's delay diff requirement. Second, given a delay upper bound $U$ and a delay diff $\delta$, we solve the Srlg-disjoint DRCR problem to find an active and a backup path at the same time such that both paths' end-to-end delays are no larger than $U$ and these two paths' end-to-end delay diff is no larger than $\delta$. Notably, both of the DRCR and the Srlg-disjoint DRCR problems are NP-Complete.

The main challenges of the DRCR and the Srlg-disjoint DRCR problems come from the delay lower bound introduced by DetNet's delay diff requirement. If there were no delay lower bound, the DRCR problem degenerates to the classical Delay Constrained Routing (DCR) problem. Although the DCR problem is NP-Complete~\cite{handler1980a}, many algorithms have been proposed to solve the DCR problem with efficacy and these algorithms can be generally grouped into four categories: 1) the K-Shortest-Path (KSP) approaches; 2) the Lagrangian-dual approaches~\cite{handler1980a, Beasley1989an, santos2007an}, 3) the dynamic programming approaches~\cite{Beasley1989an, dumitrescu2003improved, zhu2012a, thomas2019an}  and 4) the Pulse approaches~\cite{lozano2013on, sedeno2015an, cabrera2020an}. We tried to extend these approaches to handle the delay lower bound. Unfortunately, both the KSP and the Lagrangian-dual approaches may have to explore a large number of paths before finding a valid path, and thus can be extremely slow in practice; the dynamic programming approaches cannot avoid routing loops when a delay lower bound exists. The pulse approach is promising, but if we directly apply it to the DRCR problem, the optimal solution may be incorrectly skipped. The Srlg-disjoint DRCR problem is even more difficult than the DRCR problem. In addition to the challenges faced by the DRCR problem, the Srlg-disjoint DRCR problem may also encounter a ``trap'' problem, i.e., many active paths found do not have an Srlg-disjoint backup path. Although researchers have proposed ``conflict set'' to solve the trap problem~\cite{xu2004on, rostami2007cose}, existing conflict-set solvers cannot handle any delay constraints.

We propose Pulse$+$ and CoSE-Pulse$+$, to solve the DRCR problem and the Srlg-disjoint DRCR problem with optimality guarantee. The detailed contributions are as follows:

\noindent\textbf{1)} For the DRCR problem, we propose Pulse$+$, an algorithm derived from Pulse. We disable the ``dominance check'' pruning strategy in Pulse$+$ because it may cause sub-optimality in DRCR and then develop two pruning acceleration strategies (the default strategy is Large-Delay-First sorting).

\noindent\textbf{2)} For the Srlg-disjoint DRCR problem, we propose a conflict-set-finding algorithm, Conflict-Pulse$+$, which can properly handle delay constraints. Then, we integrate Pulse$+$ with a divide-and-conquer approach, and propose CoSE-Pulse$+$.

\noindent\textbf{3)} We generate synthetic test cases based on real Internet topologies and randomly-generated topologies with up to 10000 nodes, and evaluate Pulse$+$ and CoSE-Pulse$+$ against Delay-KSP, Cost-KSP and Lagrangian-KSP algorithms. Pulse$+$ and CoSE-Pulse$+$ can finish all the test cases within 200 milliseconds, while other algorithms fail to solve some cases in the time limit of 10 seconds. Further, as the network scale increases, the efficiency improvement of Pulse$+$ and CoSE-Pulse$+$ becomes more evident compared to other algorithms.

\emph{This paper does not raise ethical issues.}

%% file: sections/2background.tex
\section{Background}
In order to support time-sensitive applications in large IP networks, DetNet was standardized to offer a strict guarantee on end-to-end delay, delay jitter and packet loss~\cite{rfc8655}. Such DetNet could benefit a wide range of mission-critical real-time applications, such as telesurgery ~\cite{zhang2022application}, Eastern-Data-Western-Computing~\cite{li2022eastern}, etc. In this section, we give an overview of DetNet's design and pinpoint its challenges.

\subsection{Bounded Delay \& Jitter with CSQF}
Due to the non-deterministic queuing delays in Ethernet switches, today's packet-based networks fail to offer any Quality-of-Service (QoS) guarantee. To eliminate such non-determinism, the IETF DetNet working group developed the Cycle Specified Queuing and Forwarding (CSQF) mechanism~\cite{chen2019segment}. By specifying the sending cycle at each node along a path, CSQF guarantees bounded delay and jitter.

To support CSQF, all network nodes are synchronized within sub-microsecond accuracy, which can be achieved using the Precise Timing Protocol (PTP) as described in IEEE 802.1AS~\cite{8021as}. Then, the sending times of all the output interfaces at different nodes are divided into synchronized time intervals of equal length $T_c$. Each time interval is called a \emph{cycle}. Each output port of a node typically contains multiple output queues. During each cycle, only one queue is open, and all the packets in that queue will be transmitted. For each packet, CSQF leverages Segment Routing (SR) to specify its transmission cycle at each hop. When every node along the path follows the instructions carried in the packet, bounded delay and jitter are achieved.

We use an example in Fig. \ref{fig:CSQF_illustration} to illustrate the idea of CSQF. A packet $p$ from node $A$ to node $E$ takes the path $A\rightarrow B\rightarrow C\rightarrow D\rightarrow E$. The egress ports of $A,B,C$ and $D$ are time synchronized with cycle length equal to $T_c$. Assume that CSQF assigns a cycle list, $<1, 5, 7, 10>$, to the packet $p$, and the deadline of the packet $p$ is at the beginning of 
cycle $12$. Then, as long as the ``link delay + max processing delay'' at different hops are less than $<3T_c, T_c, 2T_c, T_c>$ respectively, the packet $p$ can be delivered to its destination $E$ before cycle $12$, and the maximum delay jitter at the node $E$ is ``$T_c+\text{max\_proc\_delay}-\text{min\_proc\_delay}$''.

In practice, the cycle assignment of a packet $p$ at an egress port is realized by specifying the queue that $p$ needs to enter. Each egress port contains $k\geq 3$ queues, and these queues open in a round-robin fashion. Consider a hop $AB$ in $p$'s path, and assume that the packet $p$ is scheduled at cycle $c_a$ at node $A$ and $c_b$ at node $B$. Then, the earliest possible arrival time of the packet $p$ at node $B$ is ``$c_a T_c+\text{link\_delay}+\text{min\_proc\_delay}$'', and the latest possible arrival time is ``$(c_a+1) T_c+\text{link\_delay}+\text{max\_proc\_delay}$''. To ensure that the packet $p$ is ready at node $B$ before cycle $c_b$, we must have
\begin{equation}\label{eqn:detnet1}
c_bT_c\geq (c_a+1) T_c+\text{link\_delay}+\text{max\_proc\_delay}.
\end{equation}
On the other hand, the packet $p$ cannot arrive at node $B$ too early. Otherwise, it may be transmitted $kT_c$ slots earlier.
\begin{equation}\label{eqn:detnet2}
(c_b-k+1)T_c<c_a T_c+\text{link\_delay}+\text{min\_proc\_delay}.
\end{equation}
Constraints (\ref{eqn:detnet1}) and (\ref{eqn:detnet2}) imply that
\begin{equation}\label{eqn:detnet_requirement}
\left\{\begin{aligned}
&(k-2)T_c>\text{max\_proc\_delay}-\text{min\_proc\_delay},\\
&(c_b-c_a)T_c\geq T_c+\text{link\_delay}+\text{max\_proc\_delay},\\
&(c_b-c_a)T_c\leq (k-1)T_c+\text{link\_delay}+\text{min\_proc\_delay}.
\end{aligned}\right.
\end{equation}
The first inequality of (\ref{eqn:detnet_requirement}) implies that at least three queues are required at each output interface. The second and the third inequalities of (\ref{eqn:detnet_requirement}) impose a lower bound and an upper bound for the per-hop delay. Therefore, to achieve a target end-to-end delay, the sum of a path's link delays must be confined in a range.

\begin{figure}[t]
    \centering
    \includegraphics[scale=0.5]{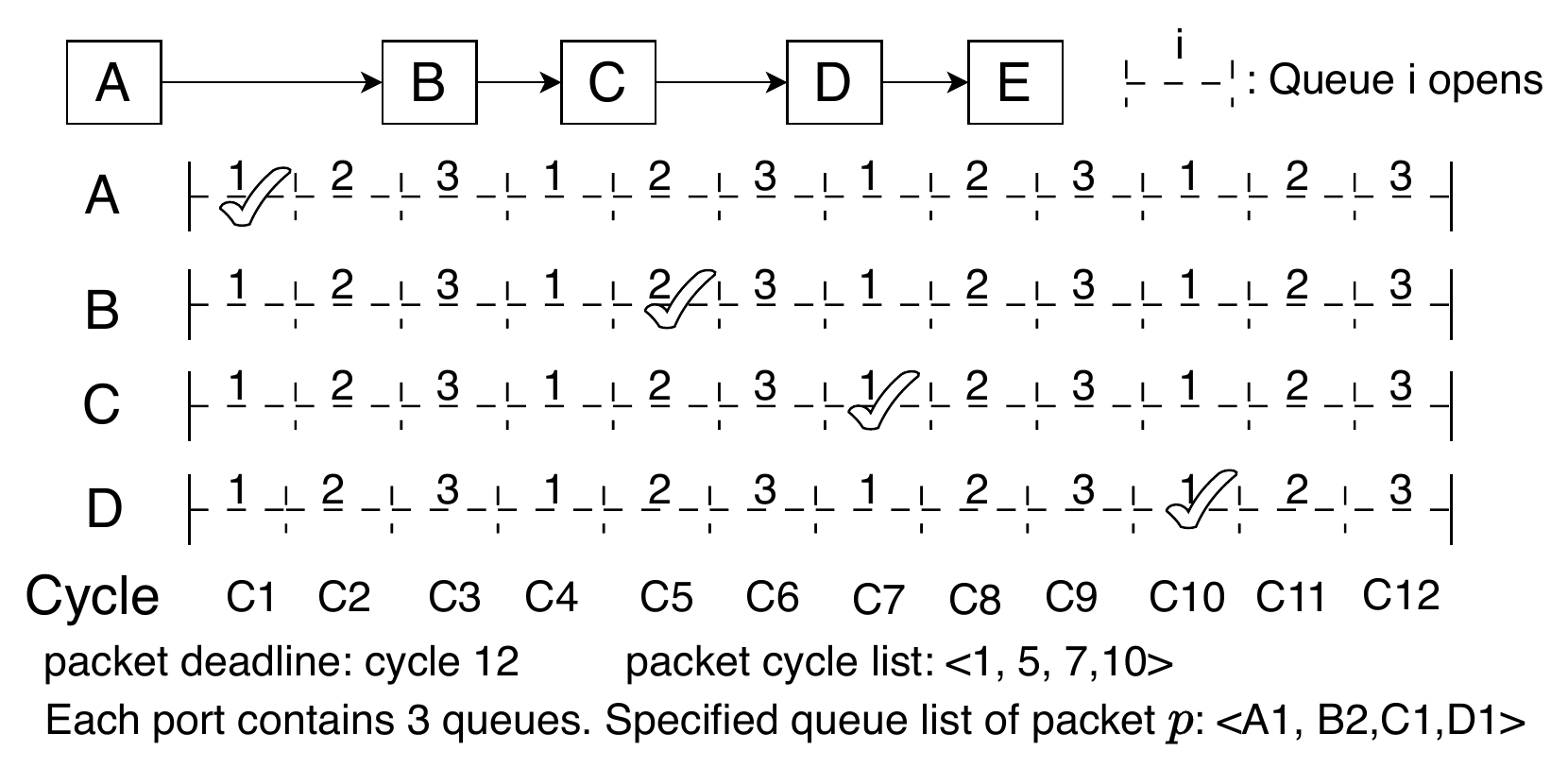}
    \caption{An example of CSQF.}
    \label{fig:CSQF_illustration}
\end{figure}

\subsection{Zero Packet Loss with 1+1 Protection}\label{sec:zero-loss}
CSQF offers bounded delay and jitter guarantee when all the network components work as expected. However, network failures are unavoidable and even a single packet loss may not be tolerable for some mission-critical applications. DetNet adopts 1+1 path protection to defend against packet loss~\cite{rfc8655}. In 1+1 protection, a backup path is used to route DetNet flow packets together with the active path. At the source node, there is a Packet Replication Function (PRF) that duplicates the received packets onto two egress ports that forward the packets to both the active path and the backup path. At the destination node, the received packets are de-duplicated using a Packet Elimination Function (PEF).

In order to guarantee bounded jitter in case of packet loss, the delay diff between the active path and the backup path must be small~\cite{sharma2022routing}. Ideally, if the received packets from both paths satisfy the ``spacing constraint'', packet recovery is easy. As shown in the lower part of Fig. \ref{fig:show_diff_need_small}, when packet 1 is lost on the active path, the receiver can recover from this packet loss by the right next packet received from the backup path. In contrast, when the ``spacing constraint'' is not met, a large packet reordering buffer would be required and a large packet reordering latency would be added to the end-to-end delay. As shown in the upper part of Fig. \ref{fig:show_diff_need_small}, when packet 1 is lost on the active path, the receiver cannot deliver the received packet 2 to the corresponding application and has to put it into its reordering buffer; only until packet 1 is received from the backup path, packet 1 and packet 2 can then be delivered. 
In this case, having a packet loss could significantly hurt the determinism of packet delivery.

\begin{figure}[t]
    \centering
    \includegraphics[scale=0.36]{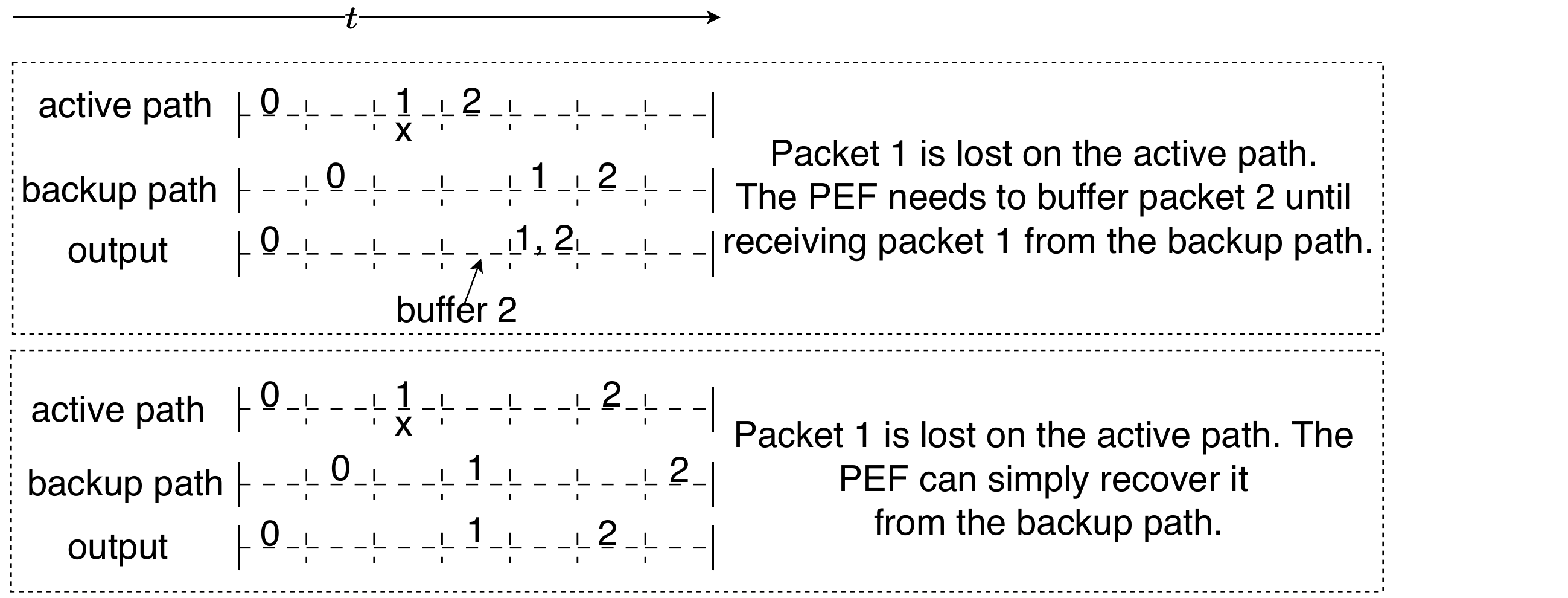}
    \caption{Delay diff constraint is critical for DetNet.}
    \label{fig:show_diff_need_small}
\end{figure}

%% file: sections/3model.tex
\section{Mathematical Model}
We model a network using a directed graph $G=(V,E)$, where $V$ is the set of nodes and $E$ is the set of directed links. Each directed link $e\in E$ is associated with a delay $d(e)$ and a cost $c(e)$. We use $\text{From}(e)$ and $\text{To(e)}$ to denote the two ends of the link $e$. Given a source node $s\in V$ and a destination node $t\in V, s\neq t$, a link sequence $P=[e_1,e_2,...,e_h]$ is called a path from $s$ to $t$ if and only if $\text{From}(e_1)=s, \text{To}(e_1) = \text{From}(e_2), ..., \text{To}(e_{h-1}) = \text{From}(e_h), \text{To}(e_h) = v$. A path $P$ is called elementary if no vertex is repeated in the path. We are interested in finding elementary paths to avoid routing loops. (The IP based forwarding will fail if a path contains a loop.) The cost and delay of a path $P$ are denoted by $d(P)=\sum_{e\in P} d(e)$ and $c(P)=\sum_{e\in P} c(e)$, respectively.

We use Shared risk link group (Srlg) to model network failures. Let $R$ be the set of Srlgs in the network $G=(V,E)$. Each Srlg $r\in R$ contains a set of links that share a common physical resource (cable, conduit, node, etc.). Thus, a failure of $r$ will cause all links in this Srlg fail simultaneously. Each link $e\in E$ may belong to multiple Srlgs. We use $\Omega(e)\subseteq R$ to denote the set of Srlgs that contain the link $e$. Then, for each path $P$, $\Omega(P)=\cup_{e\in P}\Omega(e)$ represents all the Srlgs that contain at least one link in $P$. To guard against network failures in DetNet, we study two problems in this paper.

\noindent\textbf{Delay-Range Constrained Routing (DRCR) Problem:} Given two distinct nodes $s,t\in V$ and a delay range $[L,U]$, find a min-cost path subject to the delay range constraint:
\begin{equation}\label{formulation:delay_range_routing}
\boxed{
\begin{aligned}
\min_{P}\hspace{2mm} & \hspace{10mm}c(P)=\sum_{e\in P} c(e),\\
\textbf{s.t.}\hspace{2mm} & P\text{ is an elementary path from $s$ to $t$,}\\
& L\leq d(P)=\sum_{e\in P} d(e)\leq U.
\end{aligned}
}
\end{equation}
This formulation applies to the scenarios where we have an active path $P_a$ and want to find an Srlg-disjoint backup path $P_b$ with end-to-end delay satisfying $d(P_a)-\delta\leq d(P_b) \leq d(P_a)+\delta$, where $\delta$ is the maximum allowable delay diff.

\noindent\textbf{Srlg-disjoint DRCR Problem:} Given two distinct nodes $s,t\in V$, a delay upper bound $U$ and a delay diff $\delta$, find a pair of Srlg-disjoint active and backup paths such that the active path has the minimum cost and the delay diff of the two paths does not exceed $\delta$, i.e.,
\begin{equation}\label{formulation:srlg_delay_range_routing}
\boxed{\begin{aligned}
\min_{P_a,P_b}\hspace{2mm} & \hspace{10mm}c(P_a)=\sum_{e\in P_a} c(e),\\
\textbf{s.t.}\hspace{2mm} & P_a, P_b\text{ are two elementary paths from $s$ to $t$,}\\
& d(P_a)\leq U, d(P_a)-\delta\leq d(P_b)\leq \min\{U, d(P_a)+\delta\},\\
& \Omega(P_a)\cap \Omega(P_b)=\emptyset.
\end{aligned}}
\end{equation}
This formulation is useful when we want to find the active path and the backup path at the same time.

\noindent\textbf{Remark on the objective function of (\ref{formulation:srlg_delay_range_routing}):} In DetNet, there are still many best effort packets, which can tolerate delay jitters and packet loss.  Since 1+1 protection is expensive as it doubles the traffic in the network and introduces extra processing cost, in practice 1+1 protection is only enabled for time-critical packets. As a result, the active path is used all the time, while the backup path is used less frequently. Therefore, we decide to optimize the cost of the active path, rather than optimizing the sum cost of both paths.

\noindent\textbf{Remark on the generality of finding Srlg-disjoint paths:} In some circumstances, one may care about finding link-disjoint or node-disjoint paths instead of Srlg-disjoint paths. We argue that finding link-disjoint or node-disjoint paths is a special case of finding Srlg-disjoint paths. Specifically, if every Srlg contains only one link, then finding Srlg-disjoint paths degenerates to finding link-disjoint paths; if every node in $G$, except the source node $s$ and the destination node $t$, corresponds to an Srlg, and each Srlg $u\in V, u\neq s,t$ contains all the ingress and egress links of $u$, then finding Srlg-disjoint paths degenerates to finding node-disjoint paths.

\noindent\textbf{Remark on Algorithmic Complexity:} Both the DRCR and the Srlg-disjoint DRCR problems are NP-Complete. By setting $L=0$, the DRCR problem degenerates to the Delay Constrained Routing (DCR) problem, which was proven to be NP-Complete in~\cite{handler1980a}. Thus, the DRCR problem is also NP-Complete. In addition, given a DCR problem instance, if we create a side link $e'$ from $s$ to $t$ with $d(e')\leq U$ and a large $c(e')>\sum_{e\in E}c(e)$, let this link $e'$ form a separate Srlg, and set $\delta=U$, then this DCR problem instance will reduce to an Srlg-disjoint DRCR problem instance. Therefore, the Srlg-disjoint DRCR problem is also NP-Complete.

Since the DRCR and the Srlg-disjoint DRCR problems are NP-Complete, it is impossible to design polynomial algorithms unless $P=NP$. The objective of this paper is thus to design computational efficient algorithms for these two problems, and demonstrate that they are empirically efficient to support large DetNets with thousands of nodes and links.

\section{Algorithm Design Principles}

\subsection{DRCR Problem}

To the best of our knowledge, there exists only one paper~\cite{ribeiro1985a} that directly studied the DRCR problem. However, the algorithm proposed in~\cite{ribeiro1985a} is merely a heuristic solution with no optimality guarantee. Nevertheless, if there were no lower-bound constraint on the end-to-end delay, the DRCR problem degenerates to the classical DCR problem\footnote{Also known as the Constrained Shortest Path (CSP) problem in literature.}. Existing solutions to the DCR problem can be grouped into $4$ categories. We examine these solutions one by one to identify promising algorithm-design directions for the DRCR problem.

\noindent\textbf{1) K-Shortest-Path (KSP) approaches~\cite{yen1971finding}:} The key idea is to examine all the paths with cost ordered from low to high, and the first path that meets the delay constraint gives the optimal solution. This approach is efficient if the KSP algorithm can terminate with a small $k$ value. However, when the delay bound is tight (i.e., $U-L$ is small), finding a path that meets the delay constraint may take a large number of iterations, which makes the KSP algorithm prohibitively expensive. (See Appendix \ref{appendix:ksp}.)

\noindent\textbf{2) Lagrangian-dual approaches~\cite{handler1980a, Beasley1989an, santos2007an}:} The key idea is to run the KSP algorithm based on a combined weight function $w_{\lambda}(e)=c(e)+\lambda d(e)$, where $c(e)$ and $d(e)$ are the delay and the cost of the link $e$. By properly choosing $\lambda$, the Lagrangian-dual approach could dramatically reduce the number of iterations required to find the optimal path. The Lagrangian-dual approach is effective in dealing with the delay upper bound. However, as we apply this approach to handle delay lower bound in the DRCR problem, we may need to use a \emph{negative} value for $\lambda$ in certain cases. When $\lambda$ is negative, the weight function $w_{\lambda}(e)$ may become negative and the KSP algorithm no longer applies. (See Appendix \ref{appendix:lagrangian}.)

\noindent\textbf{3) Dynamic programming approaches~\cite{Beasley1989an, dumitrescu2003improved, zhu2012a, thomas2019an}:} For a given destination node $t$, let $\psi(u,T)$ be the minimum cost of all the paths from $u$ to $t$ whose end-to-end delay is no larger than $T$. Then, $\psi(u,T)=\min_{e=(u,v)}\{\psi(v,T-d(e))+c(e)\}$. Then, starting from $\psi(t,0)=0$, we can compute each $\psi(u,T)$ and the corresponding min-cost path from $u$ to $t$ using dynamic programming. When there is no delay lower bound, all the min-cost paths found must be elementary, i.e., every node is visited at most once. Otherwise, by removing a cycle from the resulting path, a lower cost path can be found. Unfortunately, when a delay lower-bound exists, such approaches cannot guarantee the optimal path to be elementary. Hence, we decide not to pursue this direction.

\noindent\textbf{4) Pulse approaches~\cite{lozano2013on, sedeno2015an, cabrera2020an}: } These approaches use depth first search or KSP search to find a solution to the DCR problem, and adopts several pruning strategies to accelerate the search. Such approaches are the most efficient in solving the DCR problems among all the approaches. However, when a delay lower bound exists, some pruning strategies in Pulse no longer work, which reduces the pruning efficiency. Nevertheless, Pulse offers a promising framework for solving DRCR problems, and the challenge is to develop new optimization techniques to improve the pruning efficiency.

\subsection{Srlg-disjoint DRCR Problem}\label{sec:understand_srlg_drcr}

To the best of our knowledge, finding Srlg-disjoint path pairs with delay requirements has never been studied before. Nevertheless, if we remove the delay constraint, the degenerated problem did receive much attention in the past decades. We examine different solutions to find the promising algorithm-design directions and identify the corresponding challenges.

\noindent\textbf{1) Active-Path-First approaches~\cite{xu2002an, li2002efficient}:} The APF approaches first compute an active path without considering the need to find a backup path, and then try to find an Srlg-disjoint backup path by removing those links affected by the active path. If there exists no Srlg-disjoint backup path, the APF approaches may try a different active path or stop based on certain criterion. In this paper, we tried two APF approaches, one uses the KSP algorithm to find active paths (Appendix \ref{appendix:ksp}) and another one uses the Lagrangian-dual algorithm to find active paths (Appendix \ref{appendix:lagrangian}). Despite of the simplicity of the APF approaches, they may suffer from the so-called \emph{trap} problem~\cite{xu2004on}, i.e., many active paths do not have an Srlg-disjoint backup path due to some special network structure (see an example in Section \ref{sec:conflict_set}) and blindly trying different active paths can be highly inefficient. 

\noindent\textbf{2) Conflict-Set based approaches~\cite{xu2004on, rostami2007cose, xie2018divide}:} The concept of ``conflict set'' was proposed in \cite{xu2004on} to solve the \emph{trap} problem. Given an active path $P_a$, if there exists no Srlg-disjoint backup path, one can always find a small Srlg set $T\subseteq \Omega(P_a)$, such that every active path whose Srlg set contains $T$ does not have an Srlg-disjoint backup path. This set $T$ is called a ``conflict set''. If we could avoid finding active paths whose Srlg set contains a conflict set, then it would be much easier to find an Srlg-disjoint backup path. Here, the key is to compute the conflict set. Unfortunately, existing solutions~\cite{xu2004on, rostami2007cose, xie2018divide} only focused on the unconstrained routing scenarios without delay constraints, and thus cannot be used to find conflict sets for the Srlg-disjoint DRCR problem.

%% file: sections/4algorithm.tex
\section{DRCR Algorithm}\label{sec:solve_drcr}
We propose Pulse$+$ to solve the DRCR problem in this section. Let $P_{s\rightarrow t}^{\text{min\_delay}}$ be the elementary path from $s$ to $t$ with the minimum delay and let $P_{s\rightarrow t}^{\text{min\_cost}}$ be the elementary path from $s$ to $t$ with the minimum cost. Clearly, the end-to-end delay of the first path is no larger than that of the second path, i.e., $d(P_{s\rightarrow t}^{\text{min\_delay}})\leq d(P_{s\rightarrow t}^{\text{min\_cost}})$. 

All the DRCR problems can be grouped into the following six cases according to the relationship between the delay upper bound $U$, delay lower bound $L$, the min-delay path's delay $d(P_{s\rightarrow t}^{\text{min\_delay}})$ and the min-cost path's delay $d(P_{s\rightarrow t}^{\text{min\_cost}})$:

\noindent\textbf{Case 1 (Infeasible):}$L \leq U < d(P_{s\rightarrow t}^{\text{min\_delay}}) \leq d(P_{s\rightarrow t}^{\text{min\_cost}})$. It is impossible to find a path with delay smaller than the minimum delay $d(P_{s\rightarrow t}^{\text{min\_delay}})$.

\noindent\textbf{Case 2 (Degenerated Case):}  $L \leq d(P_{s\rightarrow t}^{\text{min\_delay}}) \leq U < d(P_{s\rightarrow t}^{\text{min\_cost}})$. All paths can meet the delay lower bound. Thus, the delay lower bound can be ignored and this case can be solved by the original Pulse algorithm~\cite{lozano2013on}.

\noindent\textbf{Case 3 (Trivial):}  $L \leq d(P_{s\rightarrow t}^{\text{min\_delay}}) \leq d(P_{s\rightarrow t}^{\text{min\_cost}}) \leq U$. The min-cost path $P_{s\rightarrow t}^{\text{min\_cost}}$ is the optimal solution.

\noindent\textbf{Case 4 (Non-trivial):}  $d(P_{s\rightarrow t}^{\text{min\_delay}}) < L < U < d(P_{s\rightarrow t}^{\text{min\_cost}})$. 

\noindent\textbf{Case 5 (Trivial):} $d(P_{s\rightarrow t}^{\text{min\_delay}}) < L < d(P_{s\rightarrow t}^{\text{min\_cost}}) < U$. The min-cost path $P_{s\rightarrow t}^{\text{min\_cost}}$ is the optimal solution.

\noindent\textbf{Case 6 (Non-trivial):}  $d(P_{s\rightarrow t}^{\text{min\_delay}}) < d(P_{s\rightarrow t}^{\text{min\_cost}}) < L < U$.

In this section, we first review the Pulse algorithm for the degenerated case, and then propose our algorithm to solve the two non-trivial cases.

\subsection{Review of the Pulse Algorithm}
In Case 2, the DRCR problem degenerates to the Delay Constrained Routing (DCR) problem:
\begin{equation}\label{formulation:delay_ub_routing}
\boxed{
\begin{aligned}
\min_{P}\hspace{2mm} & \hspace{10mm}c(P)=\sum_{e\in P} c(e)\\
\textbf{s.t.}\hspace{2mm} & P\text{ is an elementary path from $s$ to $t$,}\\
& d(P)=\sum_{e\in P} d(e)\leq U.
\end{aligned}
}
\end{equation}

The pulse algorithm (see Algorithm \ref{algorithm:pulse}) adopts a branch-and-bound method to find the optimal solution of (\ref{formulation:delay_ub_routing}). It defines global variables $\text{tmp\_min\_cost}$ and $P_{s\rightarrow t}^{\text{opt}}$ to track the best path found, and then performs depth first search using a stack. In the depth first search, lines 5-11 check the path found and update the best path found so far; lines 12-14 adopt three pruning strategies to cut some search branches; lines 15-17 iterate through all the egress links of the node $u$ and add the new branches to the stack. The optimal path $P_{s\rightarrow t}^{\text{opt}}$ must be an elementary path. Otherwise, $P_{s\rightarrow t}^{\text{opt}}$ will contain at least one cycle, and by removing this cycle from $P_{s\rightarrow t}^{\text{opt}}$, we could obtain another path with lower end-to-end cost.

We delve into the details of the three pruning strategies below. The first strategy ``$d(P_{s\rightarrow u})+d(P_{u\rightarrow t}^{\text{min\_delay}})> U$'' prunes branches by feasibility. It indicates that it is impossible to obtain a path with end-to-end delay no larger than $U$ through this branch. The second strategy ``$c(P_{s\rightarrow u})+c(P_{u\rightarrow t}^{\text{min\_cost}})\geq \text{tmp\_min}$'' prunes branches by optimality. It indicates that it is impossible to obtain a path with lower cost through this branch. The third strategy ``$CheckDominance(u, P_{s\rightarrow u}) == true$'' prunes branches by dominance. Given two paths $P_{s\rightarrow u}^1$, $P_{s\rightarrow u}^2$ from $s$ to $u$, $P_{s\rightarrow u}^1$ dominates $P_{s\rightarrow u}^2$ if and only if $d(P_{s\rightarrow u}^1)\leq d(P_{s\rightarrow u}^2)$ and $c(P_{s\rightarrow u}^1)\leq c(P_{s\rightarrow u}^2)$. Then, if we have searched the branch $P_{s\rightarrow u}^1$, searching the branch $P_{s\rightarrow u}^2$ cannot yield a better solution and thus can be skipped.

\begin{algorithm}[t!]
\SetAlgoLined
\KwData{A network $G(V,E)$, a source node $s$, a destination node $t$, and a delay upper bound $U$.}

\KwResult{The optimal path $P_{s\rightarrow t}^{\text{opt}}$ from $s$ to $t$.}

Use $\text{tmp\_min}$ and $P_{s\rightarrow t}^{\text{opt}}$ to track the best path found. Initialize $\text{tmp\_min} = +\infty$.

Use a stack $S$ to store all the branches to be explored. Initialize $S=\{\text{empty\_path}\}$.

\tcp{Use deep first search to find $P_{s\rightarrow t}^{\text{opt}}$.}
\While{$S$ is not empty} {
    Let path $P_{s\rightarrow u}=S.\text{pop()}$. Let $u$ be the end node of $P_{s\rightarrow u}$. Set $u=s$ if $P_{s\rightarrow u}$ is empty.
    
    \If{$u == t$} {
        \If{$d(P_{s\rightarrow u})\leq U$ and $c(P_{s\rightarrow u})<\text{tmp\_min}$} {
            $\text{tmp\_min}=c(P_{s\rightarrow u});$
            
            $P_{s\rightarrow t}^{\text{opt}}=P_{s\rightarrow u};$
        }
        \textbf{continue};
    }

    \tcp{Cut branches when possible.}   
    \If{$d(P_{s\rightarrow u})+d(P_{u\rightarrow t}^{\text{min\_delay}})> U$ or
        $c(P_{s\rightarrow u})+c(P_{u\rightarrow t}^{\text{min\_cost}})\geq \text{tmp\_min}$ or
        $CheckDominance(u, P_{s\rightarrow u}) == true$}{
        \textbf{continue};
    }

    \tcp{Add new branches.}
    \For{every egress link $e$ of the node $u$} {
        $S.\text{push}(P_{s\rightarrow u}\cup \{e\})$;
    }
}

return $P_{s\rightarrow t}^{\text{opt}}$;
\caption{Pulse Algorithm~\cite{lozano2013on}}
\label{algorithm:pulse}
\end{algorithm}

\subsection{Pulse$+$: Handling the Delay Range}
We propose Pulse$+$, an enhanced Pulse algorithm, to compute the optimal solutions for the general DRCR problems. In this section, we detail the key difficulties encountered and the optimization techniques proposed for Pulse$+$. (We also studied the KSP-based approach and the Lagrangian-Dual based approach in this paper. Since these two approaches are less efficient than Pulse$+$, we put the detailed design in Appendix \ref{appendix:other_approaches} for reference.)

\subsubsection{Dominance Check is Unsafe}\label{subsection:dominance_check_unsafe}
The efficiency of the Pulse-like algorithms heavily relies on the pruning strategies. The original Pulse algorithm adopts three pruning strategies, i.e., ``$d(P_{s\rightarrow u})+d(P_{u\rightarrow t}^{\text{min\_delay}})> U$'', ``$c(P_{s\rightarrow u})+c(P_{u\rightarrow t}^{\text{min\_cost}})\geq\text{tmp\_min}$'' and ``$CheckDominance(u, P_{s\rightarrow u}) == true$''. The first two pruning strategies are still valid, but the third one may prune a branch incorrectly for the DRCR problem and result in a sub-optimal solution.

We use two examples in Figure \ref{fig:dominance_check_example} to demonstrate the incorrectness of the Dominance check strategy. In the two examples, we need to find a min-cost path from $A$ to $E$, such that the end-to-end delay is 8 (or the delay range is $[8,8]$). Suppose that we have explored the branch $P_1 = A\rightarrow D\rightarrow C$, and we are to examine the path $P_2 = A\rightarrow B\rightarrow C$. In Figure \ref{fig:dominance_check_example}(a), $d(P_1)=3<4=d(P_2),c(P_1)=3<4=c(P_2)$, and thus $P_2$ will be pruned by the dominance check. Clearly, after pruning $P_2$, we can no longer find a path from $A$ to $E$ that meets the delay range constraint. Note that the path $A\rightarrow B \rightarrow C\rightarrow E$ meets the end-to-end delay requirement.

\begin{figure}[t]
    \centering
    \includegraphics[scale=0.7]{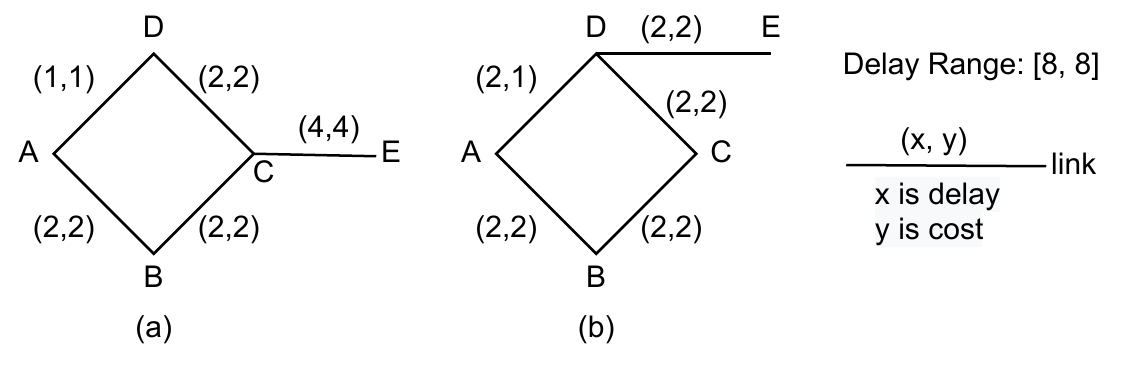}
    \caption{Dominance check is unsafe for DRCR.}
    \label{fig:dominance_check_example}
\end{figure}

The example in Figure \ref{fig:dominance_check_example}(a) hints us to modify the dominance check as follows. Given two paths $P_{s\rightarrow u}^1, P_{s\rightarrow u}^2$ from $s$ to $u$, $P_{s\rightarrow u}^1$ dominates $P_{s\rightarrow u}^2$ if and only if $d(P_{s\rightarrow u}^1)=d(P_{s\rightarrow u}^2)$ and $c(P_{s\rightarrow u}^1)\leq c(P_{s\rightarrow u}^2)$. Unfortunately, this modified dominance check is still incorrect. Consider the example in Figure \ref{fig:dominance_check_example}(b). Suppose that we are to examine $P_2 = A\rightarrow B\rightarrow C$ after exploring $P_1 = A\rightarrow D\rightarrow C$. Since $d(P_1)=4=d(P_2),c(P_1)=3<4=c(P_2)$, $P_2$ will be pruned by the dominance check, and then we can no longer find the optimal solution  $A\rightarrow B\rightarrow C\rightarrow D\rightarrow E$. Apparently, this optimal solution is attained by concatenating $P_2$ and $C\rightarrow D\rightarrow E$. However, $C\rightarrow D\rightarrow E$ cannot be concatenated with $P_1$ because the node $D$ has already been visited by $P_1$.

Admittedly, if the two paths $P_{s\rightarrow u}^1$ and $P_{s\rightarrow u}^2$ contain the same set of nodes and satisfy $d(P_{s\rightarrow u}^1)=d(P_{s\rightarrow u}^2), c(P_{s\rightarrow u}^1)\leq c(P_{s\rightarrow u}^2)$, then $P_{s\rightarrow u}^1$ will dominate $P_{s\rightarrow u}^2$. However, this pruning strategy requires memorizing (delay, cost) pairs for all the visited node sets and the total number of different node sets grows exponentially with respect to the network size, making the algorithm not scale to large networks.

Based on the above considerations, we decide to remove the ``dominance check'' pruning strategy in the Pulse$+$ algorithm. Thus, the detailed pruning strategy of Pulse$+$ (see the box in line 15 of Algorithm \ref{algorithm:pulse_plus}) becomes
\begin{equation}\label{eqn:prune_strategy}
\boxed{\begin{aligned}
&d(P_{s\rightarrow u})+d(P_{u\rightarrow t}^{\text{min\_delay}})> U&\\
\text{or}\hspace{2mm} &c(P_{s\rightarrow u})+c(P_{u\rightarrow t}^{\text{min\_cost}})\geq \text{tmp\_min}&
\end{aligned}}
\end{equation}

\begin{algorithm}[t]
\SetAlgoLined
\KwData{A network $G(V,E)$, a source node $s$, a destination node $t$, and a delay range $[L,U]$.}

\KwResult{The optimal path $P_{s\rightarrow t}^{\text{opt}}$ from $s$ to $t$.}

Use $\text{tmp\_min}$ and $P_{s\rightarrow t}^{\text{opt}}$ to track the best path found. Initialize $\text{tmp\_min} = +\infty$.

\tcp{Sort egress links to accelerate Pulse$+$.}

For every node $v\in V$, sort all the egress links of $v$ from lowest to highest based on the weight $w(e)=d(e)+d(P_{\text{To}(e)\rightarrow t}^{\text{min\_delay}})$.

\tcp{Use depth  first search to find $P_{s\rightarrow t}^{\text{opt}}$.}
Use a stack $S$ to store all the branches to be explored. Initialize $S=\{\text{empty\_path}\}$.

\While{$S$ is not empty} {
    Let path $P_{s\rightarrow u}=S.\text{pop()}$. Let $u$ be the end node of $P_{s\rightarrow u}$. Set $u=s$ if $P_{s\rightarrow u}$ is empty.
    
    \If{$u == t$} {
        \tcp{Validate the path found.}
        \If{$\boxed{L\leq d(P_{s\rightarrow u})\leq U}$} {
            \If{$c(P_{s\rightarrow u})<\text{tmp\_min}$}{
                $\text{tmp\_min}=c(P_{s\rightarrow u});$
                
                $P_{s\rightarrow t}^{\text{opt}}=P_{s\rightarrow u};$
            }
        }
        \textbf{continue};
    }

    \tcp{Cut branches when possible.}
    \If{$\boxed{P_{s\rightarrow u} \text{ should be pruned}}$}{
        \textbf{continue};
    }

    \tcp{Add new branches.}
    \For{every egress link $e$ of the node $u$} {
        \If{the node $\text{To}(e)$ is not visited in $P_{s\rightarrow u}$} {
            $S.\text{push}(P_{s\rightarrow u}\cup \{e\})$;
        }
    }
}

return $P_{s\rightarrow t}^{\text{opt}}$;
\caption{Pulse$+$ Algorithm}
\label{algorithm:pulse_plus}
\end{algorithm}

\subsubsection{Visited Node Tracking is Necessary}
Unlike the DCR problem, given a path $P_{s\rightarrow t}$ with duplicated nodes and $L\leq d(P_{s\rightarrow t})\leq U$, we cannot remove cycles from $P_{s\rightarrow t}$ to obtain a lower-cost path, as the resulting path may violate the delay lower bound. As a result, if we do not enforce that each node can only be visited once, the resulting optimal path may not be an elementary path. Take Figure \ref{fig:dominance_check_example}(b) for example. If we allow visiting a node more than once, the optimal solution would be $A\rightarrow D\rightarrow C\rightarrow D\rightarrow E$, which has an end-to-end cost of $7$. In contrast, the optimal elementary path is $A\rightarrow B\rightarrow C\rightarrow D\rightarrow E$, which has an end-to-end cost of $8$. According to the above analysis, we decide to explicitly track the visited nodes and make sure that no node is visited more than once (see line 19 in Algorithm \ref{algorithm:pulse_plus}).

\subsubsection{Largest-Delay-First Searching Strategy}\label{section:ldf}
Having removed the ``dominance check'' pruning strategy, the pruning efficiency can be impaired dramatically. We thus propose the Largest-Delay-First (LDF) Searching strategy to improve the pruning efficiency for Pulse$+$. At the beginning, since $\text{tmp\_min}=+\infty$, we can only rely on the feasibility pruning strategy ``$d(P_{s\rightarrow u})+d(P_{u\rightarrow t}^{\text{min\_delay}})> U$'' to cut branches. The LDF searching strategy explores egress links with higher end-to-end delay to the destination ($w(e)=d(e)+d(P_{\text{To}(e)\rightarrow t}^{\text{min\_delay}})$) first. The high-priority branches in the Pulse$+$ search either can be cut by the feasibility pruning strategy, or yield paths with end-to-end delay close to the delay upper bound $U$. As a result, the $\text{tmp\_min}$ value can be effectively reduced in the early stages of the Pulse$+$ DFS search, and then the optimality pruning strategy ``$c(P_{s\rightarrow u})+c(P_{u\rightarrow t}^{\text{min\_cost}})\geq \text{tmp\_min}$'' becomes more effective. Note that we use a stack to perform DFS, and a stack is last-in-first-out. Hence, to implement LDF, we need to sort all the egress links of a node $u\in V$ in an increasing order of the end-to-end delay from a link $e$ to the destination node $t$ (see line 2 in Algorithm \ref{algorithm:pulse_plus}).

We use a randomly selected test case to illustrate why LDF searching strategy could accelerate pulse$+$ search. This test case is generated in a network with 4000 nodes and 99779 links. In order to quantify the progress of Pulse$+$ search, we introduce a new concept called \emph{searching space size} ($S^3$) for every partial path $P_{s\rightarrow u}$ in the stack $S$ (see line 3 in Algorithm \ref{algorithm:pulse_plus}). The first partial path in $S$ is an empty path. An empty path means that Pulse$+$ needs to explore the whole searching space. Therefore, we set $S^3(\text{empty\_path})=1$. Every partial path $P_{s\rightarrow u}$ may generate a number of sub-paths in lines 18-21 of Algorithm \ref{algorithm:pulse_plus}. We set $S^3(P_{s\rightarrow u}\cup \{e\})=S^3(P_{s\rightarrow u})/n$, where $n$ is the number of sub-paths of $P_{s\rightarrow u}$. We say $P_{s\rightarrow u}$ is explored if and only if all of its sub-paths are explored. In Fig. \ref{fig:search_space_compare}, we plot the total searched space size of all the explored partial paths versus the number of iterations of the \textbf{while} loop (lines 4-23 in Algorithm \ref{algorithm:pulse_plus}). We can see that the searched space size increases much faster after enabling the LDF searching strategy. As a result, Pulse$+$ with LDF requires fewer number of iterations to find the optimal solution (see Fig. \ref{fig:cost_compare}).

We generate DRCR test cases (see Section \ref{sec:drcr_cases}), each of which belongs to either Case 4 or Case 6. For each test case, we record the number of iterations in Pulse$+$ search and summarize the percentile values in Table \ref{tab:apparoch reduce iteration number for DRCR}. We can see that 
enabling LDF reduces the number of iterations consistently.
\begin{figure}[t]
    \centering
    \subfigure[Searched space size vs. iteration]{
    \label{fig:search_space_compare}
    \includegraphics[scale=0.155]{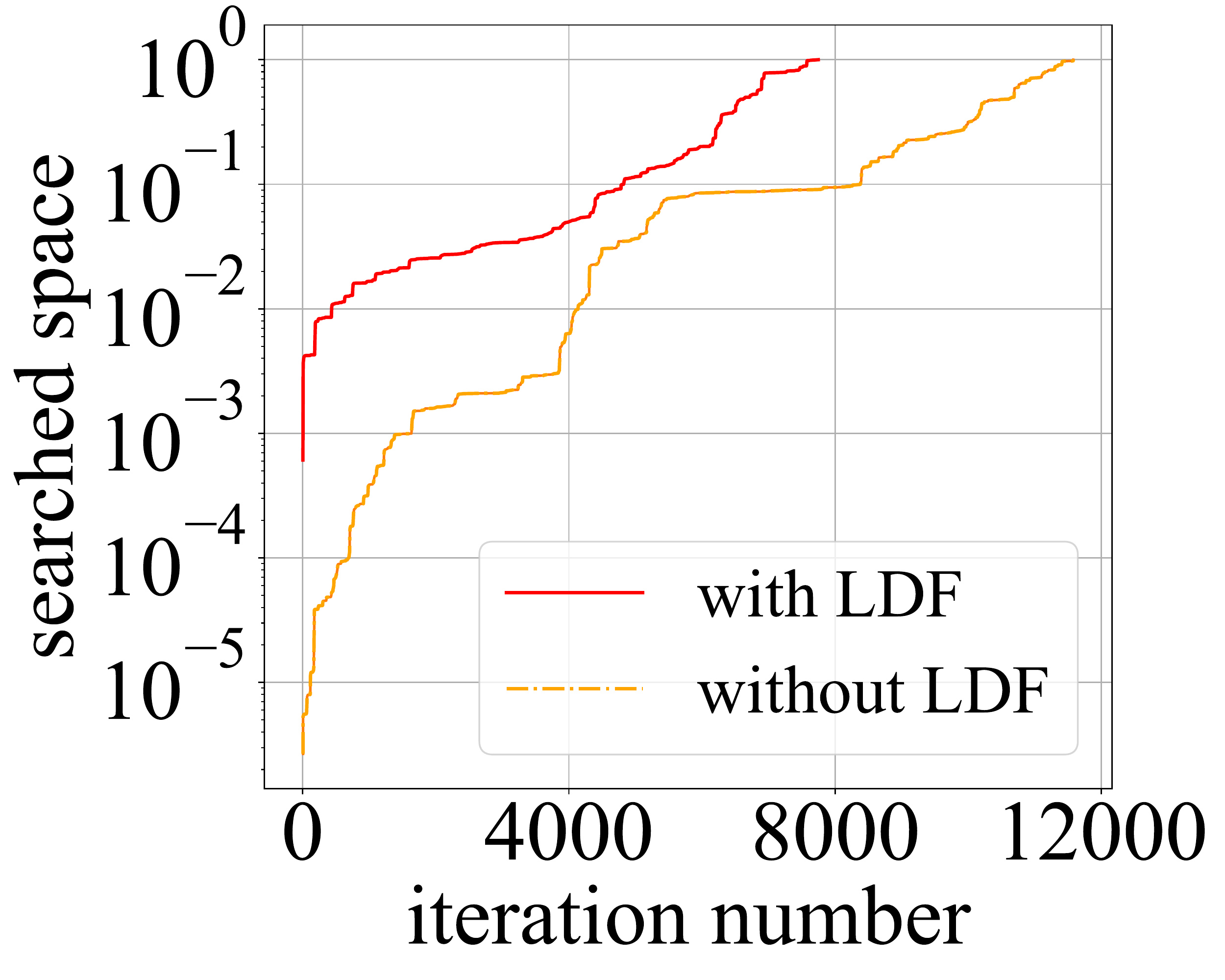}
    }
    \subfigure[Cost of best path vs. iteration]{
    \label{fig:cost_compare}
    \includegraphics[scale=0.155]{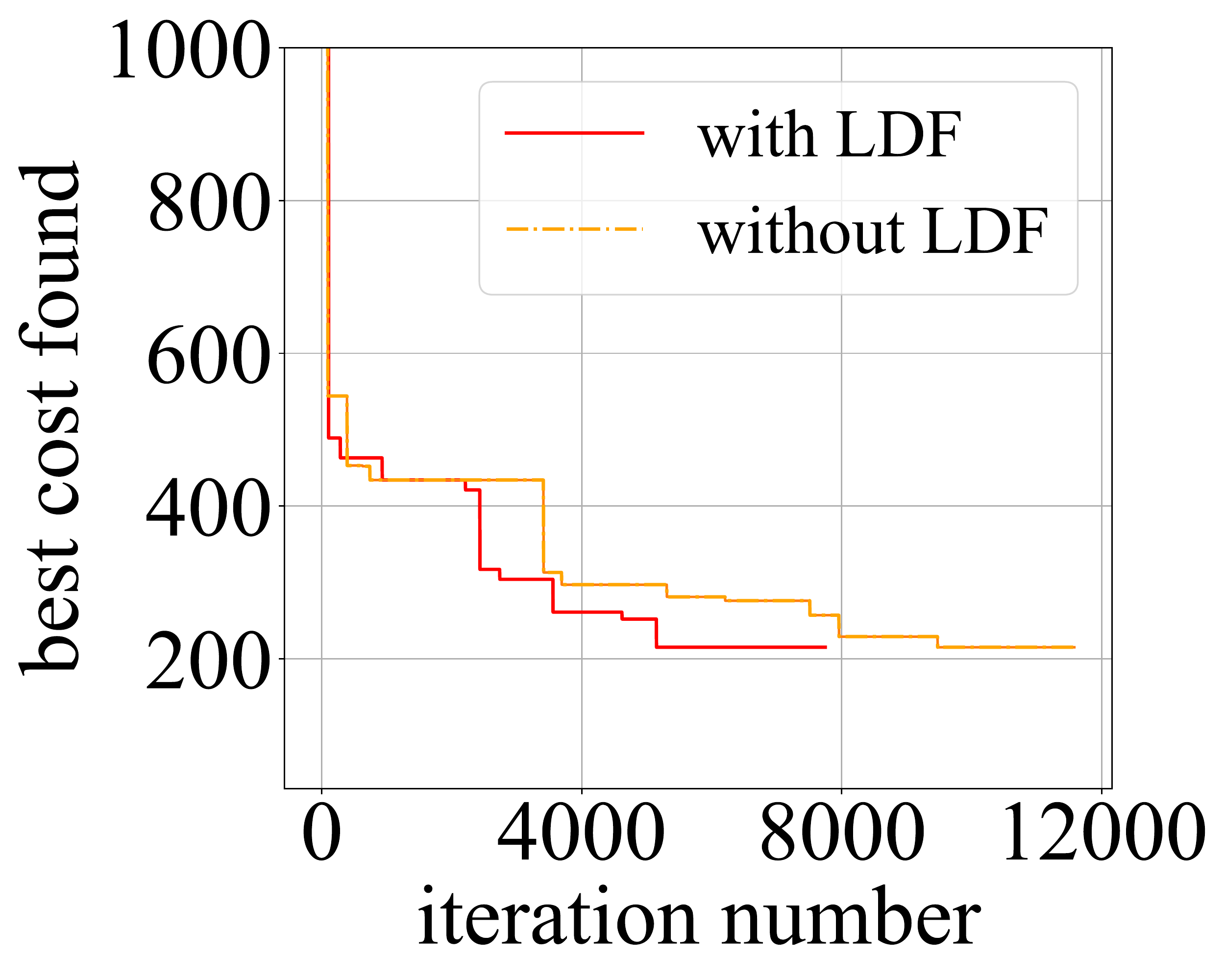}
    }
    \caption{LDF accelerates Pulse$+$ search.}
\end{figure}

\begin{table}[t]
\resizebox{\columnwidth}{!}{
\begin{tabular}{lllll}
\hline
                                                         & \textbf{DRCR Case}          & \textbf{without LDF} & \textbf{with LDF} & \textbf{Joint-Pruning} \\ \hline
\multicolumn{1}{l|}{\multirow{2}{*}{\textbf{50th Pct.}}} & \multicolumn{1}{l|}{Case 4} & 14199              & 8323                & 1606                   \\ \cline{2-5} 
\multicolumn{1}{l|}{}                                    & \multicolumn{1}{l|}{Case 6} & 46413             & 26719                & 6811                   \\ \hline
\multicolumn{1}{l|}{\multirow{2}{*}{\textbf{75th Pct.}}} & \multicolumn{1}{l|}{Case 4} & 43528             & 25218                & 2610                   \\ \cline{2-5} 
\multicolumn{1}{l|}{}                                    & \multicolumn{1}{l|}{Case 6} & 131327             & 76254               & 12043                  \\ \hline
\multicolumn{1}{l|}{\multirow{2}{*}{\textbf{99th Pct.}}} & \multicolumn{1}{l|}{Case 4} & 817173            & 489734               & 8392                   \\ \cline{2-5} 
\multicolumn{1}{l|}{}                                    & \multicolumn{1}{l|}{Case 6} & 1912911           & 1242840               & 57716                  \\ \hline
\end{tabular}
}
\caption{Strategies to reduce the number of iterations.}
\label{tab:apparoch reduce iteration number for DRCR}
\end{table}

\noindent\textbf{Another approach to accelerate Pules$+$:} LDF is not the only approach to accelerate Pules$+$. In Appendix \ref{appendix:joint_pruning}, we offer a joint-pruning approach, which could achieve even higher pruning and searching efficiency than LDF. However, the joint-pruning approach requires calculating a cost function beforehand, which incurs significant overhead. (For each test case, this overhead accounts for nearly 90\% of the total computation time.) After weighing the pros and cons, we set LDF as the default search acceleration strategy for Pulse$+$. 

\subsubsection{Optimality Guarantee of Pulse$+$}
\begin{theorem}\label{thm:optimality_pulse}
For any DRCR problem instance, as long as Pulse$+$ returns a solution, this solution must be optimal.
\end{theorem}

\begin{proof}
See Appendix \ref{appendix:proof_pulse}.
\end{proof}

\noindent\textbf{Remark:} Although Pulse$+$ guarantees optimality, its worst-case running time is not polynomial. Despite of that, thanks to the high pruning efficiency, Pulse$+$ attains much higher efficiency than other approaches including KSP and Lagrangian-dual approaches, which makes it possible to support DetNet routing in large networks with thousands of nodes and links.

%% file: sections/5disjoint.tex
\section{Srlg-Disjoint DRCR Algorithm}
We propose CoSE-Pulse$+$, to solve the Srlg-disjoint DRCR problem in this section. As discussed in Section \ref{sec:understand_srlg_drcr}, the key to CoSE-Pulse$+$ is the design of a conflict-set finding algorithm subject to delay constraints, which is described first below.



\begin{algorithm}[t]
\SetAlgoLined
\KwData{A network $G(V,E)$, a source-destination pair $(s,t)$, a delay upper bound $U$ and a path $P_a$.}

\KwResult{A conflict Srlg set $T$.}

Initialize the conflict Srlg set $T=\emptyset$.


\tcp{Perform deep first search.}
Use a stack $S$ to store all the branches to be explored. Initialize $S=\{\text{empty\_path}\}$.

\While{$S$ is not empty} {
    Let path $P_{s\rightarrow u}=S.\text{pop()}$. Let $u$ be the end node of $P_{s\rightarrow u}$. Set $u=s$ if $P_{s\rightarrow u}$ is empty.
    
    \If{there exists a disabled link in $P_{s\rightarrow u}$}{
        \textbf{continue};
    }
    \If{$u == t$} {
        \tcp{Validate the path found.}
        \If{$d(P_{s\rightarrow u})\leq U$} {
            \If{$\Omega(P_{s\rightarrow u})\cap \Omega(P_a)=\emptyset$}{
                \tcp{Fail to find a conflict set.}
                return an empty set;
            }
            \text{Pick an Srlg } $r\in \Omega(P_{s\rightarrow u})\cap \Omega(P_a)$\text{ such} \text{that $r$ contains the largest number of links.}
            
            Disable all the links in the Srlg $r$.
            
            $T.\text{insert}(r);$
        }
        \textbf{continue};
    }

    \tcp{Cut branches when possible.}
    \If{$d(P_{s\rightarrow u})+d(P_{u\rightarrow t}^{\text{min\_delay}})> U$}{
        \textbf{continue};
    }

    \tcp{Add new branches.}
    \For{every egress \textbf{active} link $e$ of the node $u$} {
        \If{the node $\text{To}(e)$ is not visited in $P_{s\rightarrow u}$} {
            $S.\text{push}(P_{s\rightarrow u}\cup \{e\})$;
        }
    }
}

return the conflict set $T$;
\caption{Conflict-Pulse$+$}
\label{algorithm:conflict_pulse_plus}
\end{algorithm}

\subsection{Conflict-Pulse$+$: Find a Conflict Set}\label{sec:conflict_set}

\subsubsection{Why do We Need Conflict Sets?}

\begin{definition}\label{def:conflict_set} (Conflict Set) Given a path $P_a$, its conflict set $T$ is a subset of $\Omega(P_a)$ such that every path $P$ whose Srlg set $\Omega(P)$ contains $T$ cannot find an Srlg-disjoint backup path.  
\end{definition}


The concept of \emph{conflict (Srlg) set} was proposed to solve the "trap" problem encountered in the link/Srlg-disjoint path finding problems, especially when the delay diff is small (which is common in DetNet). When trap happens, we get "trapped" in an infeasible solution space and cannot step out without tremendous searching.

\begin{figure}[h]
    \centering
    \includegraphics[scale=0.56]{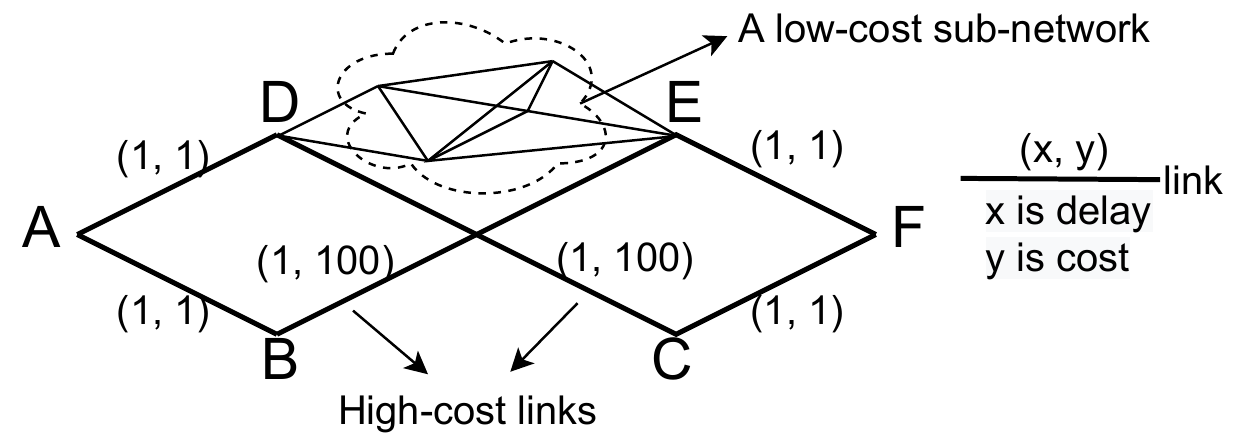}
    \caption{The Trap Problem and the Conflict Srlg Set.}
    \label{fig:trap_problem}
\end{figure}

Figure \ref{fig:trap_problem} shows an example of the trap problem. In this example, each Srlg only contains one link and thus we can use a link to represent an Srlg. The objective is to find two Srlg-disjoint paths from $A$ to $F$ such that the active path attains the minimum cost. One natural idea is to find a sequence of active paths with end-to-end cost sorted from low to high, and test if it is possible to find an Srlg-disjoint backup path. However, this approach can be extremely inefficient for the example in Figure \ref{fig:trap_problem}. Note that the links $CD$ and $BE$ have very high cost, the low-cost paths from $A$ to $F$ would be of the form $A\rightarrow D\rightarrow E\rightarrow F$ ($D\rightarrow E$ actually consists of multiple links in the low-cost sub-network in Figure \ref{fig:trap_problem}). However, none of the paths of the form $A\rightarrow D\rightarrow E\rightarrow F$ can find an Srlg-disjoint backup path. In this example, we are "trapped" in an infeasible solution space and have to do many iterations to step out. 


In Figure \ref{fig:trap_problem}, the Srlg set $\{AD,EF\}$ forms a conflict Srlg set. No path containing $\{AD,EF\}$ could find an Srlg-disjoint backup path. \textbf{Having found a number of Conflict (Srlg) Sets, if we could avoid finding an active path $P_a$ such that $\Omega(P_a)$ contains a conflict set, we could avoid the ``trap'' and accelerate the search of a feasible pair of active and backup paths.}

\subsubsection{How to Find a Conflict Set for an Active Path $P_a$?}
The problem of finding a conflict set has been studied in \cite{xu2004on, rostami2007cose, xie2018divide}. However, their approaches cannot handle delay constraints. Specifically, if an active path $P_a$ only has one backup path, but this backup path violates the delay constraint, then the existing conflict set finding algorithms in \cite{xu2004on, rostami2007cose, xie2018divide} will fail to find a conflict set, because these algorithms could incorrectly identify a backup path for $P_a$. In this section, we propose Conflict-Pulse$+$ to solve the challenge imposed by the delay constraints (see Algorithm \ref{algorithm:conflict_pulse_plus}). 

\begin{algorithm}[t]
\SetAlgoLined
\KwData{A network $G(V,E)$, a source-destination pair $(s, t)$, a delay upper bound $U$ and a delay diff $\delta$.}

\KwResult{The optimal active path $P_a^{\text{opt}}$ and an Srlg-disjoint backup path $P_b$.}

Introduce a special Srlg $r_e=\{e\}$ for each link $e$.

Define a problem instance as $I=(In, Ex)$, where $I.In$ is the set of Srlgs that must be included, and $I.Ex$ is the set of Srlgs that must be excluded.

Use $\mathcal{T}$ to denote the conflict sets found. Init $\mathcal{T}=\emptyset$.

Define a problem instance queue $Q$. Init $Q=\{(\emptyset, \emptyset)\}$.

Use $\text{tmp\_min}$ and $P_{a}^{\text{opt}}$ to track the best path found. Use $P_{b}^{\text{opt}}$ to track the backup path. Init $\text{tmp\_min} = \infty$.

\While{$Q$ is not empty} {
    Let $I=Q.\text{pop()}$;
    
    Try using the AP-Pulse$+$ algorithm to find a min-cost path $P_a$ from $s$ to $t$ such that $d(P_a)\leq U$, $I.In\subseteq\Omega(P_a)$, $I.Ex\cap\Omega(P_a)=\emptyset$, and $T\subsetneq \Omega(P_a)$ for any $T\in \mathcal{T}$.
    
    \If{$P_a$ is not found \textbf{or} $c(P_a) \geq \text{tmp\_min}$}{
        \textbf{continue};
    }
    
    Try using Pulse$+$ to find an Srlg-disjoint path $P_b$ from $s$ to $t$ such that $\Omega(P_a)\cap\Omega(P_b)=\emptyset$ and $d(P_a)-\delta\leq d(P_b)\leq \min\{U, d(P_a)+\delta\}$.
    
    \If{Pulse$+$ returns a feasible backup path $P_b$}{
        $\text{tmp\_min}=c(P_a), P_{a}^{\text{opt}}=P_a, P_{b}^{\text{opt}}=P_b$;
        
        \textbf{continue};
    }
    
    Use Conflict-Pulse$+$ to find a conflict set $T$ for $P_a$.

    \uIf{$T$ is not empty} {
        Add $T$ to $\mathcal{T}$;

        Let $\{r_1,...,r_N\}$ be the  Srlgs in $T$ but not in $I.In$;
    } \Else {
        Let $\{r_1,...,r_N\}=\{r_e:e\text{ is a link of }P_a\}$;
    }
    \For{$n=1,2,...,N$} {
        Construct a new problem instance $I_n=(I.In\cup \{r_1,r_2,...,r_{n-1}\}, I.Ex\cup \{r_n\})$;

        $Q.\text{push}(I_n)$
    }
}

Return $P_{a}^{\text{opt}}$ and $P_{b}^{\text{opt}}$.
\caption{CoSE-Pulse$+$}
\label{algorithm:cose_pulse}
\end{algorithm}

Given an active path $P_a$, we first use Pulse$+$ to check if there exists an Srlg-disjoint backup path $P_b$ satisfying $d(P_a)-\delta\leq d(P_b)\leq \min\{U, d(P_a)+\delta\}$. If not, we will run Conflict-Pulse$+$ to find a conflict set for $P_a$. In Conflict-Pulse$+$, the conflict Srlg set $T$ is initialized as an empty set (see line 1). When Conflict-Pulse$+$ finds a backup path $P_{s\rightarrow t}$ satisfying $d(P_{s\rightarrow t})\leq U$, it checks if this path $P_{s\rightarrow t}$ is an Srlg-disjoint path of $P_a$. If it is, then Conflict-Pulse$+$ fails to find a conflict set (see lines 10-12); otherwise, Conflict-Pulse$+$ picks an Srlg $r\in \Omega(P_{s\rightarrow t})\cap \Omega(P_a)$, disable all the links in $r$ and insert $r$ to $T$ (see lines 13-15). Note that, when Conflict-Pulse$+$ generates new searching branches, only \textbf{active} egress links are explored (see lines 22-26). When Conflict-Pulse$+$ encounters a branch with disabled links, it will directly cut this branch (see lines 5-7). If Conflict-Pulse$+$ can reach line 28, then the resulting Srlg set $T$ is a conflict Srlg set. This is guaranteed by the following theorem.

\begin{theorem}\label{thm:conflict_set}
Given an active path $P_a$ and any Srlg selection strategy adopted in line 16 of Algorithm \ref{algorithm:conflict_pulse_plus}, if Algorithm \ref{algorithm:conflict_pulse_plus} reaches line 28, the resulting set $T$ must be a conflict set.
\end{theorem}

\begin{proof}
See Appendix \ref{appendix:proof_conflict}.
\end{proof}

\subsection{CoSE-Pulse$+$: Solve Srlg-Disjoint DRCR}\label{sec:cose-pulse+}
Based on the concept of the conflict (Srlg) set, we propose CoSE (\textbf{Co}nflict \textbf{S}rlg \textbf{E}xclusion)-Pulse$+$ to find Srlg-disjoint paths with delay constraints. CoSE-Pulse$+$ adopts a similar divide-and-conquer approach as CoSE~\cite{rostami2007cose}. The key difference is that CoSE uses the shortest path algorithms, e.g., Dijkstra~\cite{dijkstra1959a}, $A^*$~\cite{hart1968a}, etc., to compute active/backup paths, while CoSE-Pulse$+$ uses variants of the Pulse$+$ algorithm to compute active/backup paths that meet the delay constraints and conflict sets to avoid the ``trap'' problem.

CoSE-Pulse$+$ defines a sequence of sub-problem instances $I=(In, Ex)$, where $I.In$ is the set of Srlgs that must be included, and $I.Ex$ is the set of Srlgs that must be excluded (see line 1 in Algorithm \ref{algorithm:cose_pulse}). Then, the original problem is the sub-problem $I=\{\emptyset, \emptyset\}$. Starting from each sub-problem, CoSE-Pulse$+$ first uses \textbf{AP-Pulse$+$} to find an active path $P_a$. If there exists a backup path $P_b$ for $P_a$, CoSE-Pulse$+$ updates the best path pair found so far. Otherwise, CoSE-Pulse$+$ computes a conflict Srlg set $T$ and uses this set to create new problem instances (see lines 17-27 in Algorithm \ref{algorithm:cose_pulse}). More specifically, let $\{r_1,r_2,...,r_N\}$ be the set of Srlgs in $T$ but not in $I.In$. Since $\{r_1,r_2,...,r_N\}\subseteq T\subseteq \Omega(P_a)$ and $I.Ex\cap\Omega(P_a)=\emptyset$, we must have $I.Ex\cap\{r_1,r_2,...,r_N\}=\emptyset$. Then, we can divide the problem instance $I=(In,Ex)$ into $I_1=(I.In,I.Ex\cup\{r_1\})$ and $I_1^{'}=(I.In\cup\{r_1\},I.Ex)$; $I_1^{'}$ can be further divided into $I_2=(I.In\cup\{r_1\},I.Ex\cup\{r_2\})$ and $I_2^{'}=(I.In\cup\{r_1,r_2\},I.Ex)$; $I_2^{'}$ can be further divided into $I_3=(I.In\cup\{r_1,r_2\},I.Ex\cup\{r_3\})$ and $I_3^{'}=(I.In\cup\{r_1,r_2,r_3\},I.Ex)$; and so on. Note that $I_N^{'}=(I.In\cup\{r_1,r_2,...,r_N\},I.Ex)$ is an infeasible instance, because the conflict set $T\subseteq I_N^{'}.In$. Hence, we obtain a total of $N$ sub-instances $I_1,I_2,...,I_N$ for the problem instance $I$. Note that there is a corner case where Conflict-Pulse$+$ fails to compute a conflict set. In this case, we simply use a trivial conflict set, which contains all the links of $P_a$ (see line 22). After exploring all the problem instances in $Q$, CoSE-Pulse$+$ either finds an optimal Srlg-disjoint path pair, or concludes that such an Srlg-disjoint path pair does not exist.

\subsubsection{Optimality Guarantee of CoSE-Pulse$+$}
\begin{theorem}\label{thm:optimality_cose_pulse}
For any Srlg-Disjoint DRCR problem instance, if CoSE-Pulse$+$ returns a solution, this solution must be optimal.
\end{theorem}

\begin{proof}
See Appendix \ref{appendix:proof_cose_pulse}.
\end{proof}

%% file: sections/6evaluation.tex
\section{Evaluation}\label{section:evaluation}


\subsection{Generate Test Problem Instances}\label{sec:generate_test_cases}
\subsubsection{DRCR Cases:}\label{sec:drcr_cases} We focus on the two non-trivial cases (see Section \ref{sec:solve_drcr}) when generating DRCR test cases.


\noindent\textbf{Generate Topologies:} We do not find any open source data for DetNet topologies. Instead, we use the topologies in Topology Zoo~\cite{TheInter74}, an ongoing project to collect data network topologies from all over the world. Up to now, Topology Zoo contains hundreds of different topologies, and we pick 7 largest topologies, Cogentco, GtsCe, Interoute, Kdl, Pern, TataNld and VtlWavenet2008, for evaluation. However, even these largest topologies only contain hundreds of nodes and links, which are too small to represent the real Internet.

In order to test the performance of our algorithm in large scale networks, we use ER random graph model $G(V,p)$ to generate topologies. The parameter $V$ in the model represents the number of nodes in the generated graph, and $p$ represents the probability of generating edges between two random nodes. We generate random graphs with different number of nodes and different edge connection probabilities. To characterize the influence of topology size on algorithm performance, we generate different node scales: 1) $|V|=1000$; 2) $|V|=2000$; 3) $|V| = 4000$; 4) $|V|=6000$; 5) $|V|=8000$; 6) $|V| = 10000$.  For each node scale, we use different edge connection probabilities: 1) $|p|=\ln{|V|}/|V|$; 2) $|p|=2\ln{|V|}/|V|$ 3) $|p|=3\ln{|V|}/|V|$. We choose $\ln{|V|}/|V|$ based on the conclusion that when the connection probabilities in an ER random graph is greater than $\ln{|V|}/|V|$, this graph is connected with probability 1. In order to avoid the influence of randomness on the experimental results, we generate 10 topologies for any given values of $|V|$ and $p$.

\noindent\textbf{Generate Source-destination Pairs and Delay Ranges:} For each topology, we randomly select a number of connected source-destination pairs $(s, t)$. For each pair $(s, t)$, we use Dijkstra algorithm to compute the min-delay path $P_{s\rightarrow t}^{\text{min\_delay}}$ and the min-cost path $P_{s\rightarrow t}^{\text{min\_cost}}$, and then assign different delay ranges $[L,U]$ randomly to form problem instances that belong to either of the two non-trivial cases. To meet DetNet's routing requirement, we set a upper bound on $U-L$, which is 20 in this paper.

\subsubsection{Srlg-Disjoint DRCR Cases:} \label{section:trap_proportion} We generate both trap cases and non-trap cases below.

\noindent\textbf{Generate Topologies}: We use the topologies in DRCR for the Srlg-disjoint DRCR problem and add Srlgs to the links. We add Srlgs in two styles: the star style and the non-star style~\cite{xie2018divide}. The star style is generally applied in optical networks while the non-star style can be used in other forms of network, such as the overlay network. We adopt different strategies to generate 
Srlgs of the two forms. For the star style, we randomly select the egress links of a node to be in a Srlg and the size of a Srlg is randomly determined based on the average degree of the topology. For the non-star style, we randomly select links in all the links to be in a srlg until every link is in at least one Srlg. The size of each Srlg is a random number in a given range, e.g. $[1,40]$ in our implementation.

\noindent\textbf{Generate Source-destinations Pairs and Delay Ranges}: For each topology, we randomly select a number of connected source-destination pairs $(s,t)$. For each pair $(s,t)$, we use the Dijkstra algorithm to compute the min-delay path $P_{s\rightarrow t}^{\text{min\_delay}}$, and assign the delay upper bound as $U=2.5d(P_{s\rightarrow t}^{\text{min\_delay}})$. Then we use CoSE-Pulse+ to test whether the test instance has a feasible solution (Other algorithms, such as KSP, may run indefinitely when a problem instance does not have a feasible solution.) and classify the test cases into trap and non-trap scenarios.

\noindent \textbf{Trap problem in Srlg-disjoint DRCR.} We conduct experiments to test the probability of encountering traps in Srlg-Disjoint DRCR problems. As shown in Figure \ref{fig:trap_proportion}, the trap probability increases as the delay diff decreases. Recall from Section \ref{sec:zero-loss} that the active and backup paths in DetNet cannot have a large delay diff; otherwise the PEF may not guarantee deterministic delay in case of network failures. In our evaluation of Srlg-Disjoint DRCR problems, we set the delay diff of each flow as $1ms$. In this case, about $10\%$ of all test cases encounter trap problem. We will evaluate different algorithms for both the trap cases and the non-trap cases.



\begin{figure}[h]
    \centering
    \includegraphics[scale=0.16]{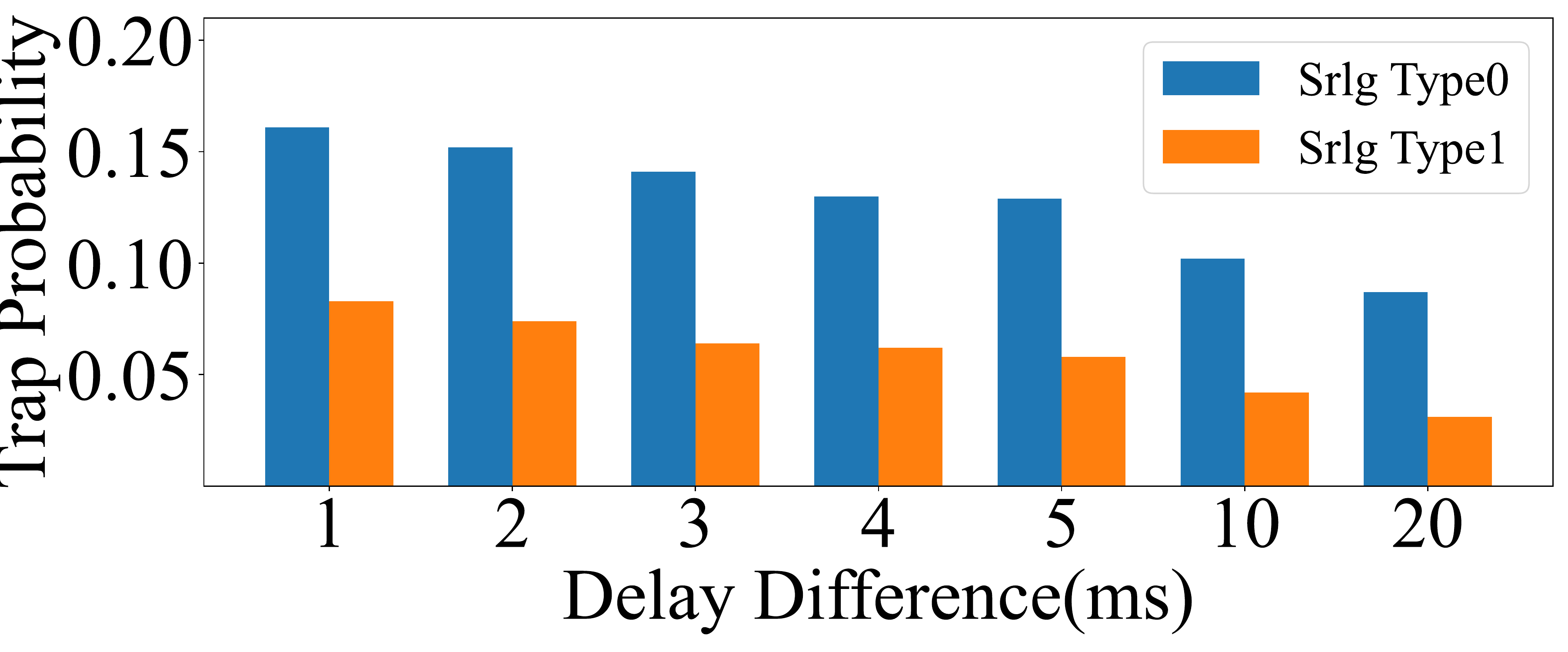}
    \caption{Trap Probability in Srlg-disjoint DRCR.}
    \label{fig:trap_proportion}
\end{figure}

\subsection{Solving DRCR Problems}
We compare Pulse$+$ with another three algorithms designed for the DRCR problem: 1) Cost-based KSP (see Appendix \ref{appendix:ksp}), 2) Lagrangian-Dual based KSP (see Appendix \ref{appendix:lagrangian}) and 3) Delay-based KSP (see Appendix \ref{appendix:ksp}). We have tried our best to optimize the code for all these algorithms to improve efficiency (see our source code in \cite{repository}). We use these these algorithms to solve all the problem instances. All experiments use a single thread of the AMD Ryzen 5600 @3.60GHz CPU on an Ubuntu workstation. Note that different test cases have different running times. Here, we calculate the percentile values and summarize the results in Table \ref{table:drcr_evaluation}. We can see that for the test cases in Topology Zoo, Pulse$+$ runs much faster than other algorithms, with over $10\times$ reduction in running time; for the test cases in random graphs, Pulse$+$ can finish all the test cases within 200 milliseconds, while other algorithms cannot with a time limit of 10 seconds. In Table \ref{tab:completion rate}, we summarize the completion rate versus the network size for all the algorithms. We can see that as the network size increases, the advantage of Pulse$+$ becomes more evident than other algorithms.


\begin{table}[t]
\resizebox{\columnwidth}{!}{
\begin{tabular}{llllll}
\toprule 
                                                         & \textbf{Topology}               & \textbf{Pulse$+$}   & \textbf{CostKsp}                  & \textbf{DelayKsp}                 & \textbf{LagrangianKsp}            \\ \midrule
\multicolumn{1}{l|}{\multirow{7}{*}{{\begin{tabular}[c]{@{}c@{}} \textbf{50th} \\ \textbf{Pct.} \end{tabular}}}} & \multicolumn{1}{l|}{Zoo}        & 70.0    & 229.0                    & 158.0                    & 188.0                    \\ \cline{2-6} 
\multicolumn{1}{l|}{}                                    & \multicolumn{1}{l|}{Node:1000}  & 1575.0   &  4537.5                   & \textgreater{}10000000.0 & 1497.0                   \\ \cline{2-6} 
\multicolumn{1}{l|}{}                                    & \multicolumn{1}{l|}{Node:2000}  & 2499.5   &  7893.0                   & \textgreater{}10000000.0 & 3802.5                   \\ \cline{2-6} 
\multicolumn{1}{l|}{}                                    & \multicolumn{1}{l|}{Node:4000}  & 7902.0   &  56172.0                  & \textgreater{}10000000.0 & 11637.0                  \\ \cline{2-6} 
\multicolumn{1}{l|}{}                                    & \multicolumn{1}{l|}{Node:6000}  & 11834.5  & 34866.0                  & \textgreater{}10000000.0 & 13972.0                  \\ \cline{2-6} 
\multicolumn{1}{l|}{}                                    & \multicolumn{1}{l|}{Node:8000}  & 18409.5  &  26379.0                  & \textgreater{}10000000.0 & 19974.5                  \\ \cline{2-6} 
\multicolumn{1}{l|}{}                                    & \multicolumn{1}{l|}{Node:10000} & 24836.0  &  82157.0                  & \textgreater{}10000000.0 & 21473.0                  \\ \hline
\multicolumn{1}{l|}{\multirow{7}{*}{{\begin{tabular}[c]{@{}c@{}} \textbf{75th} \\ \textbf{Pct.} \end{tabular}}}} & \multicolumn{1}{l|}{Zoo}        & 107.5         & 705.0                    & 397.0                    & 367.5                    \\ \cline{2-6} 
\multicolumn{1}{l|}{}                                    & \multicolumn{1}{l|}{Node:1000}  & 2185.5   &  286447.5                 & \textgreater{}10000000.0 & 3152.0                   \\ \cline{2-6} 
\multicolumn{1}{l|}{}                                    & \multicolumn{1}{l|}{Node:2000}  & 3850.5   &  449127.5                 & \textgreater{}10000000.0 & 7411.75                  \\ \cline{2-6} 
\multicolumn{1}{l|}{}                                    & \multicolumn{1}{l|}{Node:4000}  & 10256.0  & \textgreater{}10000000.0 & \textgreater{}10000000.0 & 18347.5                  \\ \cline{2-6} 
\multicolumn{1}{l|}{}                                    & \multicolumn{1}{l|}{Node:6000}  & 18019.0  &  \textgreater{}10000000.0 & \textgreater{}10000000.0 & 22437.75                 \\ \cline{2-6} 
\multicolumn{1}{l|}{}                                    & \multicolumn{1}{l|}{Node:8000}  & 27006.25 &  2836656.75               & \textgreater{}10000000.0 & 30902.0                  \\ \cline{2-6} 
\multicolumn{1}{l|}{}                                    & \multicolumn{1}{l|}{Node:10000} & 37649.5  &  \textgreater{}10000000.0 & \textgreater{}10000000.0 & 32492.0                  \\ \hline
\multicolumn{1}{l|}{\multirow{7}{*}{{\begin{tabular}[c]{@{}c@{}} \textbf{99th} \\ \textbf{Pct.} \end{tabular}}}} & \multicolumn{1}{l|}{Zoo}        & 1460.0       & \textgreater{}10000000.0 & 77430.0                  & 25054.0                  \\ \cline{2-6} 
\multicolumn{1}{l|}{}                                    & \multicolumn{1}{l|}{Node:1000}  & 15208.0  &  \textgreater{}10000000.0 & \textgreater{}10000000.0 & \textgreater{}10000000.0 \\ \cline{2-6} 
\multicolumn{1}{l|}{}                                    & \multicolumn{1}{l|}{Node:2000}  & 12341.0  &  \textgreater{}10000000.0 & \textgreater{}10000000.0 & \textgreater{}10000000.0 \\ \cline{2-6} 
\multicolumn{1}{l|}{}                                    & \multicolumn{1}{l|}{Node:4000}  & 27544.0  &  \textgreater{}10000000.0 & \textgreater{}10000000.0 & \textgreater{}10000000.0 \\ \cline{2-6} 
\multicolumn{1}{l|}{}                                    & \multicolumn{1}{l|}{Node:6000}  & 51690.0  &  \textgreater{}10000000.0 & \textgreater{}10000000.0 & \textgreater{}10000000.0 \\ \cline{2-6} 
\multicolumn{1}{l|}{}                                    & \multicolumn{1}{l|}{Node:8000}  & 66579.0  &  \textgreater{}10000000.0 & \textgreater{}10000000.0 & \textgreater{}10000000.0 \\ \cline{2-6} 
\multicolumn{1}{l|}{}                                    & \multicolumn{1}{l|}{Node:10000} & 141648.0 &  \textgreater{}10000000.0 & \textgreater{}10000000.0 & \textgreater{}10000000.0 \\ \bottomrule 
\end{tabular}
}
\caption{Solver Running Time for DRCR problems $(\mu s)$}
\label{table:drcr_evaluation}
\end{table}

\begin{table}[t]
\resizebox{\columnwidth}{!}{
\begin{tabular}{|l|l|l|l|l|l|l|}
\hline
    & \textbf{1000}   & \textbf{2000}   & \textbf{4000}   & \textbf{6000}   & \textbf{8000}   & \textbf{10000}  \\ \hline
\textbf{Pulse$+$}        & 1.0000 & 1.0000 & 1.0000 & 1.0000 & 1.0000 & 1.0000 \\ \hline
\textbf{CostKsp}       & 0.8204 & 0.8264 & 0.7244 & 0.7403 & 0.7689 & 0.6736 \\ \hline
\textbf{LagrangianKsp} & 0.9597 & 0.9693 & 0.9548 & 0.9572 & 0.9520 & 0.9496 \\ \hline
\textbf{DelayKsp}      & 0.3720 & 0.2803 & 0.2958 & 0.1975 & 0.1481 & 0.1491 \\ \hline
\end{tabular}
}
 \caption{The completion rate for DRCR problem. The time limit of each test case is 10 seconds.}
\label{tab:completion rate}
\end{table}

\subsection{Solving Srlg-disjoint DRCR Problems}
We compare CoSE-Pulse$+$, Cost-KSP, Lagrangian-KSP (Algorithm \ref{algorithm:srlg_weight_ksp} in Appendix \ref{appendix:lagrangian}) and Delay-KSP. Again, we set a time limit of 10 seconds for each problem instance. The experiment results are summarized in Table \ref{table:srlg_drcr_evaluation} and Table \ref{tab:srlg completion rate}. Cose-Pulse$+$ can finish all the test cases within 10 milliseconds, while the completion rates of all the other algorithms decrease as the network scale increases. Even though Lagrangian-KSP could performs better than Cose-Pulse$+$ for many easy cases (e.g., the non-trap cases), it fails to solve many difficult cases within the time limit. 

\begin{table}[t]
\resizebox{\columnwidth}{!}{
\begin{tabular}{llllll}
\toprule 
                                                         & \textbf{Topology}               & \textbf{Cose-Pulse$+$}   & \textbf{CostKsp}                  & \textbf{DelayKsp}                 & \textbf{LagrangianKsp}            \\ \midrule
\multicolumn{1}{l|}{\multirow{7}{*}{{\begin{tabular}[c]{@{}c@{}} \textbf{50th} \\ \textbf{Pct.} \end{tabular}}}} 
&\multicolumn{1}{l|}{Zoo}        & 24.0.0    & 211.5                    & 82657.0                   & 249.5                    \\ \cline{2-6} 
\multicolumn{1}{l|}{}                                    & \multicolumn{1}{l|}{Node:1000}  & 1222.0   &  3575.5                   & 814177.0 & 586.5                   \\ \cline{2-6} 
\multicolumn{1}{l|}{}                                    & \multicolumn{1}{l|}{Node:2000}  & 2852.0   &  57085.0                   & 3143077.0 & 1424.0                   \\ \cline{2-6} 
\multicolumn{1}{l|}{}                                    & \multicolumn{1}{l|}{Node:4000}  & 8375.5   &  936955.5                  & \textgreater{}10000000.0 & 4186.5                  \\ \cline{2-6} 
\multicolumn{1}{l|}{}                                    & \multicolumn{1}{l|}{Node:6000}  & 13969.5  &  \textgreater{}10000000.0                  & \textgreater{}10000000.0 & 7824.5                  \\ \cline{2-6} 
\multicolumn{1}{l|}{}                                    & \multicolumn{1}{l|}{Node:8000}  & 23605.5  &  \textgreater{}10000000.0                  & \textgreater{}10000000.0 & 11272.0                  \\ \cline{2-6} 
\multicolumn{1}{l|}{}                                    & \multicolumn{1}{l|}{Node:10000} & 28836.0  &  \textgreater{}10000000.0                 & \textgreater{}10000000.0 & 15445.0                  \\ \hline
\multicolumn{1}{l|}{\multirow{7}{*}{{\begin{tabular}[c]{@{}c@{}} \textbf{75th} \\ \textbf{Pct.} \end{tabular}}}} 
&\multicolumn{1}{l|}{Zoo}        & 30.65         & 800.75                    & 3409553.5                    & 867.25                    \\ \cline{2-6} 
\multicolumn{1}{l|}{}                                    & \multicolumn{1}{l|}{Node:1000}  & 1625.25   &  1283119.25                 & \textgreater{}10000000.0 & 785.75                   \\ \cline{2-6} 
\multicolumn{1}{l|}{}                                    & \multicolumn{1}{l|}{Node:2000}  & 3635.0  &   \textgreater{}10000000.0                 & \textgreater{}10000000.0 & 2192.0                  \\ \cline{2-6} 
\multicolumn{1}{l|}{}                                    & \multicolumn{1}{l|}{Node:4000}  & 9253.0  &   \textgreater{}10000000.0 & \textgreater{}10000000.0 & 6314.75                  \\ \cline{2-6} 
\multicolumn{1}{l|}{}                                    & \multicolumn{1}{l|}{Node:6000}  & 16265.0  &  \textgreater{}10000000.0 & \textgreater{}10000000.0 & 5560437.0                \\ \cline{2-6} 
\multicolumn{1}{l|}{}                                    & \multicolumn{1}{l|}{Node:8000}  & 35142.5 &   \textgreater{}10000000.0               & \textgreater{}10000000.0 & \textgreater{}10000000.0                  \\ \cline{2-6} 
\multicolumn{1}{l|}{}                                    & \multicolumn{1}{l|}{Node:10000} & 32527.0  &  \textgreater{}10000000.0 & \textgreater{}10000000.0 & \textgreater{}10000000.0                  \\ \hline
\multicolumn{1}{l|}{\multirow{7}{*}{{\begin{tabular}[c]{@{}c@{}} \textbf{99th} \\ \textbf{Pct.} \end{tabular}}}} 
&\multicolumn{1}{l|}{Zoo}        & 165.0       &6283.0                  &  \textgreater{}10000000.0 &7382.0                  \\ \cline{2-6} 
\multicolumn{1}{l|}{}                                    & \multicolumn{1}{l|}{Node:1000}  & 3608.0  &  \textgreater{}10000000.0 & \textgreater{}10000000.0 & \textgreater{}10000000.0 \\ \cline{2-6} 
\multicolumn{1}{l|}{}                                    & \multicolumn{1}{l|}{Node:2000}  & 8636.0  &  \textgreater{}10000000.0 & \textgreater{}10000000.0 & \textgreater{}10000000.0 \\ \cline{2-6} 
\multicolumn{1}{l|}{}                                    & \multicolumn{1}{l|}{Node:4000}  & 19445.0  &  \textgreater{}10000000.0 & \textgreater{}10000000.0 & \textgreater{}10000000.0 \\ \cline{2-6} 
\multicolumn{1}{l|}{}                                    & \multicolumn{1}{l|}{Node:6000}  & 28511.0  &  \textgreater{}10000000.0 & \textgreater{}10000000.0 & \textgreater{}10000000.0 \\ \cline{2-6} 
\multicolumn{1}{l|}{}                                    & \multicolumn{1}{l|}{Node:8000}  & 54670.0  &  \textgreater{}10000000.0 & \textgreater{}10000000.0 & \textgreater{}10000000.0 \\ \cline{2-6} 
\multicolumn{1}{l|}{}                                    & \multicolumn{1}{l|}{Node:10000} & 64180.0 &  \textgreater{}10000000.0 & \textgreater{}10000000.0 & \textgreater{}10000000.0 \\ \bottomrule 
\end{tabular}
}
\caption{Solver Running Time for Srlg-disjoint DRCR problems $(\mu s)$}
\label{table:srlg_drcr_evaluation}
\end{table}

\begin{table}[t]
\resizebox{\columnwidth}{!}{
\begin{tabular}{|l|l|l|l|l|l|l|}
\hline
    & \textbf{1000}   & \textbf{2000}   & \textbf{4000}   & \textbf{6000}   & \textbf{8000}   & \textbf{10000}  \\ \hline
\textbf{Cose-Pulse$+$}        & 1.0000 & 1.0000 & 1.0000 & 1.0000 & 1.0000 & 1.0000 \\ \hline
\textbf{CostKsp}       & 0.7789 & 0.6406 & 0.5590 & 0.4629 & 0.4092 & 0.3517 \\ \hline
\textbf{LagrangianKsp} & 0.8924 & 0.8438 & 0.8085 & 0.7467 & 0.6673 & 0.6749 \\ \hline
\textbf{DelayKsp}      & 0.7104 & 0.5566 & 0.4023 & 0.2571 & 0.2390 & 0.1920 \\ \hline
\end{tabular}
}
 \caption{The completion rate for Srlg-disjoint DRCR problem. The time limit of each test case is 10 seconds.}
\label{tab:srlg completion rate}
\end{table}

%% file: sections/7conclusion.tex
\section{Discussion}
\noindent\textbf{How to schedule transmission cycles for DetNet flows?} After obtaining the routing paths for each DetNet flow, the next step is to schedule the transmission cycles along the paths. We can use a central controller to schedule transmission cycles following the design principles below. 

\noindent a) First, the packets allocated to each cycle at each output interface cannot exceed the maximum number of packets that can be sent in a cycle. Otherwise, contention and deadline miss may happen. This requirement also enforces each DetNet flow to regulate its traffic using certain rate limiting and shaping functions.

\noindent b) Second, the end-to-end delays after cycle assignments along both paths must be close to each other. Given the end-to-end link delay of a path and Equation (\ref{eqn:detnet_requirement}), we can easily obtain a range for the achievable end-to-end delay after cycle assignment. As long as the active path's achievable delay range overlaps with the backup path's achievable delay range, DetNet packets could experience very close end-to-end delays along both paths with proper cycle assignment.

\noindent\textbf{How to improve DetNet flow's admission rate?} As the number of admitted DetNet flows increases, some links may not have sufficient resources to schedule additional DetNet flows. In this case, we could increase the cost of the congested links. Then, CoSE-Pulse$+$ or Pulse$+$ will avoid these links. This simple scheme could achieve better load balance and increase DetNet flow's admission rate.

\noindent\textbf{How to avoid link under-utilization?} DetNet flows reserve transmission cycles to achieve deterministic delay and jitter. Due to the rate fluctuation, some cycles may not have enough DetNet packets to send. In this case, best-effort packets can be transmitted.

\section{Related Work}

Pulse$+$ and CoSE-Pulse$+$ meet all the routing requirements of DetNet flows in large networks with thousands of nodes and links. To the best of our knowledge, none of the existing solutions could achieve this objective.

Most works on DetNet routing and scheduling did not account for network failures~\cite{Falk2018Exploring, Nayak2018Incremental, Nayak2018Routing, Schweissguth2020ILP-Based, CHANG2021Time-predictable, Krolikowski2021joint}. The RFC standard of DetNet proposed using backup paths to protect against network failures~\cite{rfc8655}. A recent paper~\cite{sharma2022routing} formulated the active/backup path finding problem using integer programming, but the computational complexity is too high.

From the pure algorithm design's point of view, the link/Srlg-disjoint path finding problems have been studied with an objective to minimize 1) the sum cost of both paths~\cite{suurballe1984a, hu2003diverse, gomes2011resilient, bermond2015finding} or 2) the min cost of the two paths~\cite{xu2002an, li2002efficient, xu2004on, rostami2007cose, xie2018divide, vass2022polynomial}. However, none of these works could handle delay constraints.

The DRCR problem studied in this paper arises as a sub-problem of the Srlg-disjoint DRCR problem. Due to the delay diff requirement, a delay lower bound is imposed. Most existing literature on delay constrained routing does not account for the delay lower bound constraints~\cite{handler1980a, santos2007an, dumitrescu2003improved, zhu2012a, thomas2019an, lozano2013on, sedeno2015an, cabrera2020an}. Although the algorithm proposed in~\cite{ribeiro1985a} directly handles delay lower bounds, it cannot guarantee optimality.


The DRCR problem is similar to another line of research works~\cite{Desrochers1992a, feillet2004an, lozano2015an, costa2019exact}, i.e., the Vehicle Routing Problem with Time Windows (VRPTW). Given a graph $G(V,E)$, each link $e\in E$ is associated with a delay-cost pair $(d(e), c(e))$ and each node $v\in V$ is associated with a time window $[L_v, U_v]$. The objective is to deliver a service from $s$ to $t$, such that the delivery time is in $[L_t, U_t]$. Note that the service in the VRPTW problem is allowed to arrive at a node $v$ earlier than $L_v$ and then wait until $L_v$ to start its next delivery. In contrast, our DRCR problem does not allow early arrival. In DetNet, network switches may not have enough memory to buffer the early-arrival packets.

\section{Conclusion}

DetNet introduces stringent routing requirements to achieve low end-to-end delay, low delay jitter and zero packet loss. We propose Pulse$+$ and CoSE-Pulse$+$ to solve DetNet's routing challenges. Pulse$+$ and CoSE-Pulse$+$ not only have theoretical optimality guarantee, but also exhibit great scalability in empirical tests. Pulse$+$ and CoSE-Pulse$+$ make it possible to achieve fast routing computation in large-scale DetNets with thousands of nodes and links.



%% file: sections/appendix1.tex
\clearpage

\section{Other Approaches for DRCR and Srlg-Disjoint DRCR}\label{appendix:other_approaches}

\subsection{K-Shortest Path (KSP) is Too Slow}\label{appendix:ksp}
We adopt Yen's KSP algorithm~\cite{yen1971finding} here. Yen's KSP algorithm can be applied either to cost or delay. In the cost-based KSP approach (see Algorithm \ref{algorithm:cost_ksp}), the first path that meets the delay range constraint attains the optimal end-to-end cost. In the delay-based KSP approach (see Algorithm \ref{algorithm:delay_ksp}), we need to iterate over all the paths that meet the delay range constraint and pick the one with the lowest end-to-end cost. 

\noindent\textbf{Performance:} We evaluate both KSP algorithms using the DRCR test cases generated in Section \ref{sec:drcr_cases}. The percentile values of the algorithm running times are summarized in Table \ref{table:drcr_evaluation}. We can see that both Ksp algorithms fail to compute solutions for many test cases in a 10-second time limit, and the completion ratio decreases as the network size increases. As a result, KSP algorithms cannot meet the time requirement for route calculation in a large-scale DetNet.

\noindent\textbf{Solving Srlg-disjoint DRCR Problems: } The KSP-based approach can be easily generalized to solve the Srlg-disjoint DRCR problem. We only need to replace the boxed condition (line 3 in Algorithm \ref{algorithm:cost_ksp} and line 1 in Algorithm \ref{algorithm:delay_ksp}) by
$$\boxed{
\begin{aligned}
&L\leq d(P_{s\rightarrow t}^k)\leq U\text{ and There exists an Srlg-disjoint backup} \\
&\text{path for }P_{s\rightarrow t}^k\text{ and this backup path meets the delay-range}\\
&\text{and delay-diff constraints.}
\end{aligned}}$$
However, such algorithms can be extremely inefficient when solving the trap cases.

\subsection{Lagrangian Approach}\label{appendix:lagrangian}
The KSP algorithm can only start searching from the shortest path. When the delay lower bound is large, the KSP algorithm has to perform a large number of iterations before finding a feasible path that meets the delay lower bound, which makes the pure KSP-based approaches inefficient. The Lagrangian Approach can be utilized to reduce the number of iterations in the KSP search. Specifically, the Lagrangian Approach introduces a weighted cost $w_{\lambda}(e)=c(e)+\lambda d(e)$, where $\lambda$ is a real number, and then performs KSP search based on $w_{\lambda}(e)$. With a properly chosen $\lambda$, the shortest path under $w_{\lambda}$ could have an end-to-end delay close to the delay range $[L,U]$. As a result, the number of KSP iterations can be reduced.

\begin{algorithm}[t]
\SetAlgoLined
\KwData{A network $G(V,E)$, a source node $s$, a destination node $t$, and a delay range $[L,U]$.}

\KwResult{The optimal path $P_{s\rightarrow t}^{\text{opt}}$ from $s$ to $t$.}

\For{k=1,2,...} {
    Use Yen's K-Shortest Path algorithm to find the path $P_{s\rightarrow t}^k$ with the $k$-th smallest end-to-end cost;

    \If{$\boxed{L\leq d(P_{s\rightarrow t}^k)\leq U}$}{
        return $P_{s\rightarrow t}^k$ as the optimal solution of $P_{s\rightarrow t}^{\text{opt}}$.
    }
}
\caption{Cost-based K-Shortest Path}
\label{algorithm:cost_ksp}
\end{algorithm}

\begin{algorithm}[t]
\SetAlgoLined
\KwData{A network $G(V,E)$, a source node $s$, a destination node $t$, and a delay range $[L,U]$.}

\KwResult{The optimal path $P_{s\rightarrow t}^{\text{opt}}$ from $s$ to $t$.}

Use Yen's KSP algorithm to find the set $\mathcal{P}$ of all the paths $P_{s\rightarrow t}^k$'s with $\boxed{L\leq d(P_{s\rightarrow t}^k) \leq U}$;

Find the $P_{s\rightarrow t}^k\in\mathcal{P}$ with the smallest $c(P_{s\rightarrow t}^k)$, and return it as the optimal solution of $P_{s\rightarrow t}^{\text{opt}}$.

\caption{Delay-based K-Shortest Path}
\label{algorithm:delay_ksp}
\end{algorithm}

\begin{algorithm}[h]
\SetAlgoLined
\KwData{A network $G(V,E)$, a source node $s$, a destination node $t$, and a delay range $[L,U]$.}

\KwResult{The optimal path $P_{s\rightarrow t}^{\text{opt}}$ from $s$ to $t$.}

Initialize a weight upper bound $w^U=+\infty$.

Use $\text{tmp\_min}$ and $P_{s\rightarrow t}^{\text{opt}}$ to track the best path found. Initialize $\text{tmp\_min} = +\infty$.

\For{k=1,2,...} {
    Use Yen's K-Shortest Path algorithm to find the path $P_{s\rightarrow t}^k$ with the $k$-th smallest end-to-end weight $w_{\lambda}(P_{s\rightarrow t}^k)=\sum_{e\in P_{s\rightarrow t}^k}w_{\lambda}(e)$;

    \If{$P_{s\rightarrow t}^k$ is not found \textbf{or} $w_{\lambda}(P_{s\rightarrow t}^k) \geq w^U$}{
        return $P_{s\rightarrow t}^{\text{opt}}$ as the optimal solution.
    }

    \If{$\boxed{L\leq d(P_{s\rightarrow t}^k)\leq U}$}{
        \If{$c(P_{s\rightarrow t}^k)<\text{tmp\_min}$}{
            $\text{tmp\_min} = c(P_{s\rightarrow t}^k)$;
            
            $w^U= c(P_{s\rightarrow t}^k) + \max\{\lambda * L, \lambda * U\}$;
        }
    }
}
\caption{Lagrangian-based K-Shortest Path}
\label{algorithm:weight_ksp}
\end{algorithm}

\subsubsection{Finding Optimal Solution based on $w_{\lambda}$}
We first study how to perform KSP search based on $w_{\lambda}$ to obtain the optimal solution of the DRCR problem. The detailed algorithm is shown in Algorithm \ref{algorithm:weight_ksp}. The following lemma states that Algorithm \ref{algorithm:weight_ksp} gives the optimal solution.

\begin{lemma}\label{lemma:weight_ksp}
For any $\lambda$, Algorithm \ref{algorithm:weight_ksp} returns the optimal path solution to the DRCR problem.
\end{lemma}

\begin{proof}
We prove by contradiction. Let $P_{s\rightarrow t}^*$ be the path that attains the smallest end-to-end cost and meets the delay range constraint $d(P_{s\rightarrow t}^*)\in [L,U]$. If $c(P_{s\rightarrow t}^*) < c(P_{s\rightarrow t}^{\text{opt}})$, then we must have 
\begin{equation}\label{eqn:contradict}
w_{\lambda}(P_{s\rightarrow t}^*) \geq w^U.
\end{equation}
Otherwise, Algorithm \ref{algorithm:weight_ksp} would be able to find the path $P_{s\rightarrow t}^*$.

On the other hand, since $d(P_{s\rightarrow t}^*)\in [L,U]$ and $P_{s\rightarrow t}^*$ attains the minimum cost, we must have
$$
\begin{aligned}
w_{\lambda}(P_{s\rightarrow t}^*)&=c(P_{s\rightarrow t}^*)+\lambda d(P_{s\rightarrow t}^*)\\
&< c(P_{s\rightarrow t}^{\text{opt}})+\max\{\lambda * L, \lambda * U\}=w^U,
\end{aligned}
$$
which contradicts to the inequality (\ref{eqn:contradict}).
\end{proof}

\noindent\textbf{Remark:} To support Srlg-disjoint DRCR problems, we can replace the boxed condition (see line 8 in Algorithm \ref{algorithm:weight_ksp}) by $$\boxed{
\begin{aligned}
&L\leq d(P_{s\rightarrow t}^k)\leq U\text{ and There exists an Srlg-disjoint backup} \\
&\text{path $P_b$ for }P_{s\rightarrow t}^k\text{ and $P_b$ meets all the delay constraints.}
\end{aligned}
}$$
However, the above approach may perform unnecessary backup path searches. In CostKsp (see Algorithm \ref{algorithm:cost_ksp}), when we find an active path that meets the delay range constraint, we know that this path has the optimal cost. In contrast, for the $w_{\lambda}$-based KSP, the first active path found may not have the optimal cost. Even if this path has an Srlg-disjoint backup path, we cannot conclude that we find the optimal solution. To solve the above problem, we propose an enhanced $w_{\lambda}$-based KSP algorithm for finding the optimal Srlg-disjoint path pair in Algorithm \ref{algorithm:srlg_weight_ksp}. The key idea is to perform backup path search only for the paths with the optimal end-to-end cost (see lines 9-12).

\begin{algorithm}[t]
\SetAlgoLined
\KwData{A network $G(V,E)$, a source-destination pair $(s,t)$, a delay upper bound $U$ and a delay diff $\delta$.}

\KwResult{The optimal active path $P_a^{\text{opt}}$ from $s$ to $t$ and an Srlg-disjoint backup path $P_b$.}

Initialize a weight upper bound $w^U=+\infty$.

Define a min heap $H$ to track the active paths that meet the delay constraint. Initially, $H$ is empty. The active paths in $H$ are ordered by end-to-end cost.

\For{k=1,2,...} {
    Use Yen's K-Shortest Path algorithm to find the path $P_{s\rightarrow t}^k$ with the $k$-th smallest end-to-end weight $w_{\lambda}(P_{s\rightarrow t}^k)=\sum_{e\in P_{s\rightarrow t}^k}w_{\lambda}(e)$;

    \If{$P_{s\rightarrow t}^k$ is not found}{
        break;
    }

    \If{$w_{\lambda}(P_{s\rightarrow t}^k) \geq w^U$}{
        Pop the active path from the top of the min heap $H$, denoted by $P_a$.
        
        \If{$P_a$ has an Srlg-disjoint backup path $P_b$ \textbf{\emph{and}} $d(P_a)-\delta\leq d(P_b)\leq \min\{U, d(P_a)+\delta\}$} {
            return $P_a$ and $P_b$ as the optimal solution.
        }
        
        Let $\text{min\_cost}$ be the end-to-end cost of the path on the top of the min heap $H$. Set $\text{min\_cost}=\infty$ if $H$ is empty.

        Set $w^U= \text{min\_cost} + \lambda * U$;
    }

    \If{$d(P_{s\rightarrow t}^k)\leq U$}{
        Add $P_{s\rightarrow t}^k$ to the heap $H$;
        
        Let $\text{min\_cost}$ be the end-to-end cost of the path on the top of the min heap $H$.

        Set $w^U= \text{min\_cost} + \lambda * U$;
    }
}

\While{$H$ is not empty} {
    Pop the active path from the top of the min heap $H$, denoted by $P_a$.
    
    \If{$P_a$ has an Srlg-disjoint backup path $P_b$ \textbf{\emph{and}} $d(P_a)-\delta\leq d(P_b)\leq \min\{U, d(P_a)+\delta\}$} {
        return $P_a$ and $P_b$ as the optimal solution.
    }
}
\caption{Srlg-Disjoint Lagrangian-KSP}
\label{algorithm:srlg_weight_ksp}
\end{algorithm}

\subsubsection{Choosing an Appropriate $\lambda$}\label{appendix:choose_lambda_lagrangian}
We then discuss how to choose an appropriate $\lambda$ for the Lagrangian approach. We focus on the two non-trivial cases in this section: 1) $d(P_{s\rightarrow t}^{\text{min\_delay}})< L< U< d(P_{s\rightarrow t}^{\text{min\_cost}})$ and 2) $d(P_{s\rightarrow t}^{\text{min\_delay}})< d(P_{s\rightarrow t}^{\text{min\_cost}})< L< U$.

We first introduce some mathematical preliminaries. For any $\lambda$, let $P_{s\rightarrow t}^{\lambda-\text{opt}}$ be the path that attains the smallest weight $w_{\lambda}(P_{s\rightarrow t}^{\lambda-\text{opt}})=\sum_{e\in P_{s\rightarrow t}^{\lambda-\text{opt}}}w_{\lambda}(e)$. (Note that the end-to-end delay of $P_{s\rightarrow t}^{\lambda-\text{opt}}$ may not be within $[L,U]$.) We define 
\begin{equation}\label{eqn:g}
\begin{aligned}
g(\lambda)&=w_{\lambda}(P_{s\rightarrow t}^{\lambda-\text{opt}})-\max\{\lambda L, \lambda U\}\\
&=c(P_{s\rightarrow t}^{\lambda-\text{opt}})+\lambda d(P_{s\rightarrow t}^{\lambda-\text{opt}})-\max\{\lambda L, \lambda U\}.
\end{aligned}
\end{equation}
We can prove the following lemmas for $g(\lambda)$.



\begin{lemma}
$g(\lambda)$ is a concave function of $\lambda$, i.e., $g(\lambda_1)+g(\lambda_2)\leq 2g((\lambda_1+\lambda_2)/2), \text{ for any } \lambda_1, \lambda_2.$
\end{lemma}

\begin{proof}
According to $g(\lambda)$'s definition (\ref{eqn:g}), we only need to prove the following two inequalities:
$$\begin{aligned}
&w_{\lambda_1}(P_{s\rightarrow t}^{\lambda_1-\text{opt}})+w_{\lambda_2}(P_{s\rightarrow t}^{\lambda_2-\text{opt}})
\leq 2w_{(\lambda_1+\lambda_2)/2}(P_{s\rightarrow t}^{(\lambda_1+\lambda_2)/2-\text{opt}}),\\
&\max\{\lambda_1 L, \lambda_1 U\}+\max\{\lambda_2 L, \lambda_2 U\}
\geq \max\{(\lambda_1+\lambda_2)L, (\lambda_1+\lambda_2)U\}.\nonumber
\end{aligned}$$

For any path $P$, we must have
$$\begin{aligned}
&2w_{(\lambda_1+\lambda_2)/2}(P)=\sum_{e\in P}\left(2c(e)+(\lambda_1+\lambda_2)d(e)\right)\\
=&\sum_{e\in P}\left(c(e)+\lambda_1 d(e)\right)+\sum_{e\in P}\left(c(e)+\lambda_2 d(e)\right)=w_{\lambda_1}(P) + w_{\lambda_2}(P).
\end{aligned}$$
Then, the first inequality can be proved as follows:
$$\begin{aligned}
&2w_{(\lambda_1+\lambda_2)/2}(P_{s\rightarrow t}^{(\lambda_1+\lambda_2)/2-\text{opt}})\\
=&w_{\lambda_1}(P_{s\rightarrow t}^{(\lambda_1+\lambda_2)/2-\text{opt}}) + w_{\lambda_2}(P_{s\rightarrow t}^{(\lambda_1+\lambda_2)/2-\text{opt}})\\
\geq&w_{\lambda_1}(P_{s\rightarrow t}^{\lambda_1-\text{opt}})+w_{\lambda_2}(P_{s\rightarrow t}^{\lambda_2-\text{opt}}).\nonumber
\end{aligned}$$
The ``$\geq$'' holds because $P_{s\rightarrow t}^{\lambda_i-\text{opt}}$ is the min-weight path under the weight function $w_{\lambda_i}, i=1,2$.

The second inequality is equivalent to the following two inequalities:
$$\lambda_1 L + \lambda_2 L \leq \max\{\lambda_1 L, \lambda_1 U\}+\max\{\lambda_2 L, \lambda_2 U\},$$
$$\lambda_1 U + \lambda_2 U \leq \max\{\lambda_1 L, \lambda_1 U\}+\max\{\lambda_2 L, \lambda_2 U\}.$$
The above two inequalities hold obviously.
\end{proof}

\begin{lemma}
Let $P_{s\rightarrow t}^*$ be the optimal solution of the DRCR problem. For any $\lambda$, we have
$g(\lambda)\leq c(P_{s\rightarrow t}^*).$
\end{lemma}

\begin{proof}
$g(\lambda)=w_{\lambda}(P_{s\rightarrow t}^{\lambda-\text{opt}})-\max\{\lambda L, \lambda U\}\leq w_{\lambda}(P_{s\rightarrow t}^*)-$
$\max\{\lambda L, \lambda U\} = c(P_{s\rightarrow t}^*) + \lambda d(P_{s\rightarrow t}^*) - \max\{\lambda L, \lambda U\}\leq c(P_{s\rightarrow t}^*).$
\end{proof}

Now we are ready to discuss how to choose a proper $\lambda$. Our goal is to find $\lambda^*$ that maximizes $g(\lambda)$. When $\lambda = \lambda^*$, the min-weight path $P_{s\rightarrow t}^{\lambda^*-\text{opt}}$ would have its end-to-end weight $w_{\lambda}(P_{s\rightarrow t}^{\lambda^*-\text{opt}})$ close to the upper bound $w^U$ in Algorithm \ref{algorithm:weight_ksp} and thus the number of iterations required would be small.

\textbf{Case 1: $d(P_{s\rightarrow t}^{\text{min\_delay}})< L< U< d(P_{s\rightarrow t}^{\text{min\_cost}})$.} We prove that $\lambda^*>0$ in this case. Consider the right derivative $g'(0^+)$ of $g(\lambda)$ at $\lambda = 0$. When $\lambda = 0$, it is easy to see that $P_{s\rightarrow t}^{\text{min\_cost}}$ is the min-weight path. Thus, $g'(0^+)=d(P_{s\rightarrow t}^{\text{min\_cost}})-U>0$. Consider the derivative $g'(\Lambda)$ at a sufficiently large value $\Lambda$. When $\lambda=\Lambda$, the cost term ``$c(e)$'' becomes negligibly small compared to the delay term ``$\Lambda d(e)$'' in $w_{\Lambda}(e)$, and then $P_{s\rightarrow t}^{\text{min\_delay}}$ becomes the min-weight path. Thus, $g'(\Lambda)=d(P_{s\rightarrow t}^{\text{min\_delay}})-U\leq 0.$ Based on the above analysis, we know that in the range $[0,\Lambda]$, $g(\lambda)$ increases at the beginning while decreases in the end. Combined with the fact that $g(\lambda)$ is a concave function, we know that the optimal solution $\lambda^*$ must be within the range $[0,\Lambda]$. Then, we use the gradient descent algorithm to compute the optimal $\lambda^*$.

\textbf{Case 2: $d(P_{s\rightarrow t}^{\text{min\_delay}})< d(P_{s\rightarrow t}^{\text{min\_cost}})< L< U$.} We prove that $\lambda^*<0$ in this case. Consider the left derivative $g'(0^-)$ of $g(\lambda)$ at $\lambda = 0$. When $\lambda = 0$, $P_{s\rightarrow t}^{\text{min\_cost}}$ is the min-weight path. Thus, $g'(0^-)=d(P_{s\rightarrow t}^{\text{min\_cost}})-L<0$. Combined with the fact that $g(\lambda)$ is a concave function, we know that the optimal solution $\lambda^* < 0$. 

In the region $\lambda < 0$, $w_{\lambda}(e)$ may become negative. Note that in the definition of $g(\lambda)$ in (\ref{eqn:g}), calculating $g(\lambda)$ requires finding a shortest path with respect to the link weight $w_{\lambda}(e)$. If we use Dijkstra's algorithm to compute such a shortest path, all the link weights must be non-negative. In order to meet the requirement of the Dijkstra's algorithm, we define $\mu=\min_{e\in E} \{c(e)/d(e)\}$ and only consider those $\lambda$'s satisfying $\lambda \geq -\mu$. Consider the derivative of $g'(-\mu)$ of $g(\lambda)$ at $\lambda = -\mu$. Base on the value of $g'(-\mu)=d(P_{s\rightarrow t}^{-\mu-\text{opt}})-L$, we can further divide Case 2 into two subcases:

\textbf{Subcase 1:} $d(P_{s\rightarrow t}^{-\mu-\text{opt}})>L$.] In this subcase, $g'(-\mu)>0$. Combined with the fact that $g'(0^-)<0$, we know that the optimal $\lambda^* \in (-\mu, 0)$. Then, we can still use the gradient descent algorithm to compute the optimal $\lambda^*$.

\textbf{Subcase 2:} $d(P_{s\rightarrow t}^{-\mu-\text{opt}})\leq L$.] In this subcase, $g'(-\mu)\leq 0$. Thus, the optimal $\lambda^*$ must be in the range $(-\infty, -\mu]$. On the other hand, to make sure $w_{\lambda}(e)$'s are non-negative, we require $\lambda\geq -\mu$. Hence, we just set $\lambda^*=-\mu$ in this subcase.

\noindent\textbf{Remark: } In the most extreme cases where there exists a link $e$ such that $c(e)=0$ and $d(e)>0$, we would have $\mu=0$. Then, if a problem instance falls into the subcase 2, the Lagrangian approach would become completely useless, because we can only set $\lambda^*=0$.

\begin{figure}[h]
    \centering
    \includegraphics[scale=0.51]{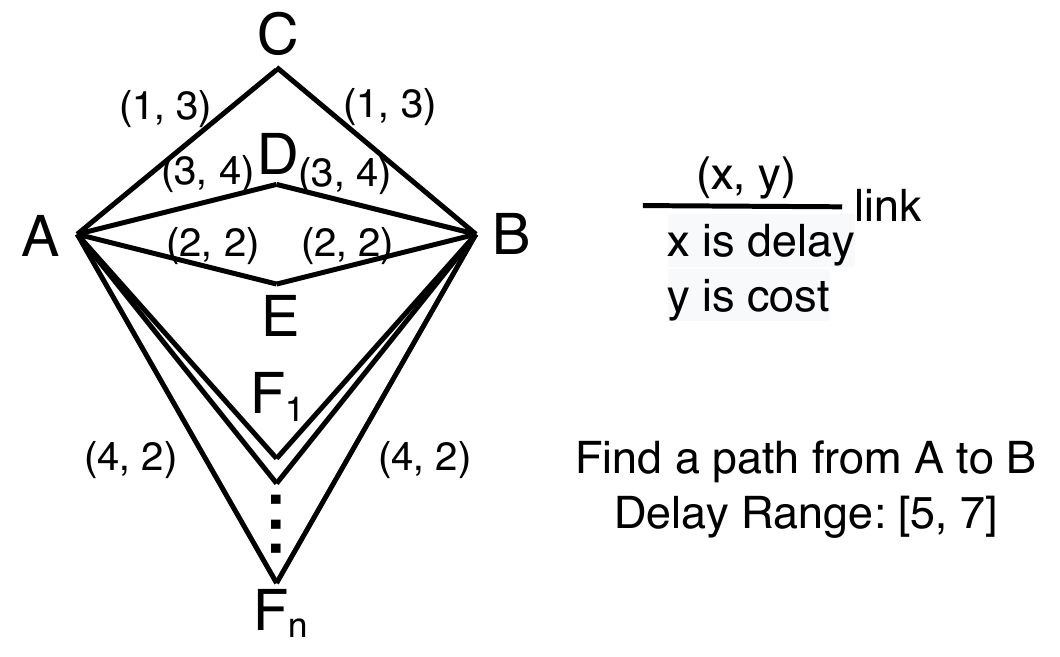}
    \caption{Example: Maximizing $g(\lambda)$ Yields More KSP Iterations. All the $n$ $A\rightarrow F_i\rightarrow B$ paths have the same end-to-end cost of $4$ and end-to-end delay of $8$. }
    \label{fig:lagrangian_not_optimal}
\end{figure}

\noindent\textbf{Remark on the Sub-Optimality of the Lagrangian Approach: } Note that we choose a $\lambda$ that maximizes $g(\lambda)$ in the Lagrangian Approach. This approach could help reduce the number of KSP iterations in many cases. However, there is no guarantee that this choice of $\lambda$ is always optimal. Consider the example in Figure \ref{fig:lagrangian_not_optimal}. The objective is to find a path from $A$ to $B$ such that the end-to-end delay is in the range $[5,7]$. Clearly, the only path that meets the delay range constraint is $A\rightarrow D\rightarrow B$, with an end-to-end delay of $3+3=6$. If we use Delay-KSP (Algorithm \ref{algorithm:delay_ksp}) or use a very large $\lambda$ in the Lagrangian Approach, 3 KSP iterations would be required ($A\rightarrow C\rightarrow B$, $A\rightarrow E\rightarrow B$, and then $A\rightarrow D\rightarrow B$). On the other hand, if we choose the $\lambda$ that maximizes $g(\lambda)$, we can write down the formula of $g(\lambda)$ as follows
\begin{equation}\label{eqn:cost_function_example}
g(\lambda)=\left\{
\begin{aligned}
&4+8\lambda-5\lambda=4+3\lambda, & \text{ if }\lambda\leq 0, \\
&4+4\lambda-7\lambda=4-3\lambda, & \text{ if }0<\lambda\leq 1, \\
&6+2\lambda-7\lambda=6-5\lambda, & \text{ if }\lambda>1,
\end{aligned}
\right.
\end{equation}
and find that $\lambda=0$ maximizes $g(\lambda)$. However, if we use $\lambda=0$ in the Lagrangian Approach, $n+3$ KSP iterations would be required ($A\rightarrow E\rightarrow B$ or $A\rightarrow F_i\rightarrow B$, then $A\rightarrow C\rightarrow B$, and finally $A\rightarrow D\rightarrow B$). Clearly, using $\lambda=0$ yields much more KSP iterations than using a very large $\lambda$. This example could explain why Lagrangian-KSP is not always better than Cost-KSP and Delay-KSP.

\noindent\textbf{Performance:} From Table \ref{table:drcr_evaluation}, we can see that the Lagrangian-KSP Approach runs much faster than both Cost-KSP and Delay-KSP, and achieves higher completion ratio in a 10-second time limit. However, the Lagrangian-KSP Approach is not always effective. In certain cases, the Lagrangian-KSP Approach still cannot compute a solution within the time limit. The reason is that, some DRCR problem instances may fall into the subcase 2. In this subcase, we can only use a suboptimal $\lambda^*=-\mu$ to perform the KSP search. Due to the large gap between $g(\lambda^*)=g(-\mu)$ and the optimal solution $c(P_{s\rightarrow t}^*)$, more KSP iterations are required on average. In addition, as shown in the example in Fig. \ref{fig:lagrangian_not_optimal}, even if the Lagrangian-KSP Approach can choose the optimal $\lambda$, there is no guarantee that the total number of iterations will be reduced.

%% file: sections/appendix2.tex
\section{Joint Pruning for Pulse$+$}\label{appendix:joint_pruning}
We propose joint pruning, another method to accelerate Pulse$+$ search. For every node $u$, we define a cost function $f_u:\mathbb{R}^+\rightarrow\mathbb{R}^+$ that maps from delay budget to path cost, and use this function to prune a branch $P_{s\rightarrow u}$ if the following condition is met:
\begin{equation}\label{eqn:joint_delay_cost_prune}
\boxed{c(P_{s\rightarrow u})+f_u(U-d(P_{s\rightarrow u}))\geq\text{tmp\_min}.}
\end{equation}
Then, we have the following theorem.

\begin{theorem}\label{lem:joint_delay_cost_prune}
Assume that the cost function $f_u(l)$ satisfies the following two conditions for all $$l=U-d(P_{s\rightarrow u})\in (-\infty, U-d(P_{s\rightarrow u}^{\text{min\_delay}})].$$
\begin{enumerate}
    \item $f_u(l)=+\infty$ if there exists no elementary path $P_{u\rightarrow t}$ from $u$ to $t$ satisfying $d(P_{u\rightarrow t})\leq l$;
    \item $f_u(l)\leq c(P_{u\rightarrow t})$ for any elementary path $P_{u\rightarrow t}$ from $u$ to $t$ satisfying $d(P_{u\rightarrow t})\leq l$.
\end{enumerate}
Then, if a branch $P_{s\rightarrow u}$ satisfies the constraint (\ref{eqn:joint_delay_cost_prune}), searching this branch will not yield a solution with end-to-end cost lower than $\text{tmp\_min}$.
\end{theorem}

\begin{proof}
The proof is straightforward. If $f_u(U-d(P_{s\rightarrow u}))=+\infty$, then (\ref{eqn:joint_delay_cost_prune}) always holds. In this case, there does not exist an elementary path $P_{u\rightarrow t}$ such that $d(P_{s\rightarrow u}) + d(P_{u\rightarrow t})\leq U$. Therefore, the branch $P_{s\rightarrow u}$ can be safely cut.

If $f_u(U-d(P_{s\rightarrow u}))<+\infty$, then for any elementary path $P_{u\rightarrow t}$ satisfying $d(P_{s\rightarrow u})+d(P_{u\rightarrow t})\leq U$, we must have $c(P_{s\rightarrow u})+c(P_{u\rightarrow t})\geq c(P_{s\rightarrow u})+f_u(U-d(P_{s\rightarrow u}))\geq\text{tmp\_min}$. In this case, searching the branch $P_{s\rightarrow u}$ will not yield a solution with end-to-end cost lower than $\text{tmp\_min}$, and thus this branch can be safely cut.
\end{proof}

The strength of the pruning strategy (\ref{eqn:joint_delay_cost_prune}) depends on the cost function $f_u$. Consider the following cost function:
\begin{equation}\label{eqn:cost_function_example}
f_u(l)=\left\{
\begin{aligned}
&+\infty, & \text{ if }l< d(P_{u\rightarrow t}^{\text{min\_delay}}), \\
&c(P_{u\rightarrow t}^{\text{min\_cost}}), & \text{ if }l\geq d(P_{u\rightarrow t}^{\text{min\_delay}}).
\end{aligned}
\right.
\end{equation}
It is easy to verify that the pruning strategy (\ref{eqn:joint_delay_cost_prune}) degenerates to the original pruning strategy (\ref{eqn:prune_strategy}).

In order to obtain better pruning effect, we need a tighter cost function. Obviously, the tightest cost function is \begin{equation}\label{eqn:opt_cost_function}
f_u(l)=\min_{P_{u\rightarrow t}\text{ is a path satisfying } d(P_{u\rightarrow t})\leq l} c(P_{u\rightarrow t}).
\end{equation}
Next, we design a pseudo-polynomial algorithm that computes the above cost function.


\subsection{Calculating the Cost Function (\ref{eqn:opt_cost_function})}\label{subsubsection:calculate_cost_function}

\begin{definition}\label{def:revised_dominance_check}
(Dominance Check) Given two partial paths $P_{u\rightarrow t}^1$ and $P_{u\rightarrow t}^2$ from $u$ to the destination node $t$, $P_{u\rightarrow t}^2$ is dominated by $P_{u\rightarrow t}^1$ if the following two conditions hold: 1) $d(P_{u\rightarrow t}^1) \leq d(P_{u\rightarrow t}^2)$; 2) $c(P_{u\rightarrow t}^1)\leq c(P_{u\rightarrow t}^2)$.
\end{definition}

Now we are ready to present the algorithm that computes the cost function (\ref{eqn:opt_cost_function}). Starting from the destination node $t$, we perform the smallest-cost-first search using a min heap in the reverse direction. We represent each searching branch by a \emph{(node, delay, cost)} tuple, and this tuple means that there exist a path $P$ (which may contain cycles) from \emph{node} to the destination node $t$ such that $d(P)=\text{delay}$ and $c(P)=\text{cost}$. In lines 6-8 of Algorithm \ref{algorithm:relax_cost_function}, we perform the dominance check, and cut a branch if its delay-cost pair is dominated by another delay-cost pair searched before. In lines 14-16 of Algorithm \ref{algorithm:relax_cost_function}, we perform the delay-based feasibility check, and cut a branch if it is not possible to obtain a path from the source node $s$ to the destination node $t$ that meets the delay upper bound. After the smallest-cost-first search, we obtain a list of delay-cost pairs for each node $u\in V$. Then, the relaxed cost function $f_u(l)$ can be calculated according to line 20 in Algorithm \ref{algorithm:relax_cost_function}. The correctness of Algorithm \ref{algorithm:relax_cost_function} is guaranteed by the following theorem.

\begin{algorithm}[t!]
\SetAlgoLined
\KwData{A network $G(V,E)$, a source node $s$, a destination node $t$, and a delay range $[L, U]$.}

\KwResult{The cost function $f_u$ for every $u\in V$.}

Define a $Z=(\text{node, delay, cost})$ tuple. Given $Z_1$ and $Z_2$, define $Z_1 < Z_2$ if and only if $Z_1.\text{cost}<Z_2.\text{cost}$.

Define a min heap $H$ for the node-delay-cost tuple. Initialize $H=\{(t, 0, 0)\}$.

For every node $u\in V$, define a delay-cost pair set $\Pi_u$. Initialize $\Pi_u=\emptyset$.

\While{$H$ is not empty} {
    Let $(u, \text{delay\_to\_dst, cost\_to\_dst})=H.\text{pop()}$;
    
    \If{there exist a delay-cost pair in $\Omega_u$ that dominates $(\text{delay\_to\_dst, cost\_to\_dst})$}{
        \textbf{continue};
    }

    $\Pi_u.\text{insert((delay\_to\_dst, cost\_to\_dst))}$;
    
    \For{every ingress link $e$ of the node $u$}{
        Let $v=\text{From}(e)$;
        
        Let $\text{new\_delay}=\text{delay\_to\_dst}+d(e)$;
        
        Let $\text{new\_cost}=\text{cost\_to\_dst}+c(e)$;
        
        \If{$\text{new\_delay}+d(P_{s\rightarrow v}^{\text{min\_delay}})\leq U$}{
            $H.\text{push}((v, \text{new\_delay}, \text{new\_cost}))$;
        }
    }
}

\For{each node $u$ and each delay value $l$}{
    Compute $f_u(l)=\min\{c: \text{there exists } d\leq l \text{ such that } (d,c)\in \Pi_u\}$;
}

\caption{Calculate $f_u$ According to (\ref{eqn:opt_cost_function})}
\label{algorithm:relax_cost_function}
\end{algorithm}

\begin{theorem}\label{thm:relax_cost_function}
For every node $u$ and every delay value $l\leq U-d(P_{s\rightarrow u}^{\text{min\_delay}})$, the $f_u(l)$ computed by Algorithm \ref{algorithm:relax_cost_function} is equal to the relaxed cost function defined by (\ref{eqn:opt_cost_function}).
\end{theorem}



\begin{proof}
When $l< d(P_{u\rightarrow t}^{\text{min\_delay}})$, there exists no path from node $u$ to node $t$. Hence, the values of $f_u(l)$ computed by Algorithm \ref{algorithm:relax_cost_function} and Equation (\ref{eqn:opt_cost_function}) are both $+\infty$.

When $d(P_{u\rightarrow t}^{\text{min\_delay}})\leq l\leq U-d(P_{s\rightarrow u}^{\text{min\_delay}})$, let $\mathcal{P}_{u \rightarrow t}$ be the set of paths that attain the smallest end-to-end cost among all paths whose end-to-end delays are within $[0, l]$. We only need to prove that there exist a path $P_{u\rightarrow t}\in \mathcal{P}_{u\rightarrow t}$ such that $(d(P_{u\rightarrow t}), c(P_{u\rightarrow t}))\in\Pi_u$.

We prove by contradiction. Assume that there exists no path $P_{u\rightarrow t}\in \mathcal{P}_{u\rightarrow t}$ such that $(d(P_{u\rightarrow t})$, $c(P_{u\rightarrow t}))\in\Pi_u$. For every path $P_{u\rightarrow t}\in \mathcal{P}_{u\rightarrow t}$, let $u_1=u, u_2, u_3,..., u_K=t$ be the sequence of nodes visited by the path $P_{u\rightarrow t}$, and let $P_{u_i\rightarrow u_j}^{\text{sub}}$ be the sub-path of $P_{u\rightarrow t}$ from $u_i$ to $u_j$. Since $(d(P_{u\rightarrow t}), c(P_{u\rightarrow t}))\notin\Pi_u$, there must exist a node $u_k$ such that the branch $P_{u_k\rightarrow t}^{\text{sub}}$ is cut by Algorithm \ref{algorithm:relax_cost_function}. We call the node $u_k$ as the \textbf{cut node} of the path $P_{u\rightarrow t}$. Among all the paths in $\mathcal{P}_{u\rightarrow t}$, we consider the path $P_{u\rightarrow t}^*$ whose \textbf{cut node} is closest to the node $u$ in terms of the hop count along the same path. We will derive contradictions based on $P_{u\rightarrow t}^*$.

\noindent\textbf{Step 1:} We first check whether the branch $P_{u_k^*\rightarrow t}^{\text{sub}}$($u_k^*$ is the \textbf{cut node} of the path $P_{u\rightarrow t}^*$) is cut by the feasibility check in lines 14-16 of Algorithm \ref{algorithm:relax_cost_function}. Since
$$\begin{aligned}
&d(P_{s\rightarrow u_k^*}^{\text{min\_delay}})+d(P_{u_k^*\rightarrow t}^{\text{sub}})
\leq d(P_{s\rightarrow u}^{\text{min\_delay}})+d(P_{u\rightarrow t}^*)\\
\leq & d(P_{s\rightarrow u}^{\text{min\_delay}})+l\leq U,\end{aligned}$$
it does not violate the feasibility check. Hence, the branch $P_{u_k^*\rightarrow t}^{\text{sub}}$ cannot be cut by the feasibility check. 

\noindent\textbf{Step 2:} We then check the possibility of this branch being cut by the dominance check in lines 6-8 of Algorithm \ref{algorithm:relax_cost_function}. Consider the delay interval $I=[0, d(P_{u_k^*\rightarrow t}^{\text{sub}})-d(P_{u\rightarrow t}^*)+l]$. Since $d(P_{u\rightarrow t}^*)\in [0, l]$, it is easy to verify that $d(P_{u_k^*\rightarrow t}^{\text{sub}})\in I$. Since the branch $P_{u_k^*\rightarrow t}^{\text{sub}}$ is cut by the dominance check, there must exist another path $P_{u_k^*\rightarrow t}^{'}$ such that $d(P_{u_k^*\rightarrow t}^{'})\in I$, $c(P_{u_k^*\rightarrow t}^{'})\leq c(P_{u_k^*\rightarrow t}^{\text{sub}})$ and $P_{u_k^*\rightarrow t}^{'}$ is not cut at the node $u_k^{*}$. Since $d(P_{u_k^*\rightarrow t}^{'})\in I$, if we concatenate the sub-path $P_{u\rightarrow u_k^*}^{\text{sub}}$ with $P_{u_k^*\rightarrow t}^{'}$ and denote this path as $P_{u\rightarrow t}^{'}$, then
$$d(P_{u\rightarrow t}^{'})=d(P_{u_k^*\rightarrow t}^{'}) + d(P_{u\rightarrow t}^*)-d(P_{u_k^*\rightarrow t}^{\text{sub}})\leq l.$$
We derive contradictions in two cases. 
\begin{description}
\item[Case 1:]$c(P_{u_k^*\rightarrow t}^{'})<c(P_{u_k^*\rightarrow t}^{\text{sub}})$. Since
$$\begin{aligned}
& c(P_{u\rightarrow t}^{'})=c(P_{u_k^*\rightarrow t}^{'})+c(P_{u\rightarrow u_k^*}^{\text{sub}})\\
<&c(P_{u_k^*\rightarrow t}^{\text{sub}})+c(P_{u\rightarrow u_k^*}^{\text{sub}})=c(P_{u\rightarrow t}^*),\end{aligned}$$
we obtain a path $P_{u\rightarrow t}^{'}$ with lower cost than those in $\mathcal{P}_{u\rightarrow t}$, which contradicts to the definition of $\mathcal{P}_{u\rightarrow t}$.
\item[Case 2:]$c(P_{u_k^*\rightarrow t}^{'})=c(P_{u_k^*\rightarrow t}^{\text{sub}})$. In this case, $c(P_{u\rightarrow t}^{'})=c(P_{u\rightarrow t}^*),$ $d(P_{u\rightarrow t}^{'})\leq l$. Thus, $P_{u\rightarrow t}^{'}\in \mathcal{P}_{u\rightarrow t}$. Now, consider the \textbf{cut node} of the path $P_{u\rightarrow t}^{'}$. Note that $P_{u\rightarrow t}^{'}$ is not cut at the node $u_k^*$, and $P_{u\rightarrow t}^{'},P_{u\rightarrow t}^{*}$ share the same sub-path from $u$ to $u_k^*$. The \textbf{cut node} of the path $P_{u\rightarrow t}^{'}$ must have a smaller hop count away from the node $u$ than the \textbf{cut node} $u_k^*$ of the path $P_{u\rightarrow t}^{*}$, which contradicts to the choice of the path $P_{u\rightarrow t}^{*}$.
\end{description}

We have proved that when $d(P_{u\rightarrow t}^{\text{min\_delay}})\leq l\leq U-d(P_{s\rightarrow u}^{\text{min\_delay}})$, there exists a path $P_{u\rightarrow t}\in \mathcal{P}_{u\rightarrow t}$ such that $(d(P_{u\rightarrow t}),$ $c(P_{u\rightarrow t}))\in\Pi_u$. Therefore, the $f_u(l)$ value calculated by Algorithm \ref{algorithm:relax_cost_function} is the same as that of Equation (\ref{eqn:opt_cost_function}).
\end{proof}

\subsection{The Overhead of Joint Pruning is Large}
From Table \ref{tab:apparoch reduce iteration number for DRCR}, we can see that the joint-pruning strategy is much more effective in reducing the number of iterations for Pulse$+$ search. But unfortunately, the joint-pruning strategy incurs significant overhead, which limits its usage in practice. We quantify the overhead of joint pruning using \emph{overhead ratio}, which is equal to the time required to compute the cost function $f_u$ divided by the total solver running time (including both the cost-function calculation time and the Pulse$+$ running time). As shown in Table \ref{tab:overhead_ratio}, the average overhead ratio reaches about 90\%.

\begin{table}[h]
\resizebox{\columnwidth}{!}{
\begin{tabular}{|l|l|l|l|l|l|l|}
\hline
\textbf{Network Scale}    & \textbf{1000}   & \textbf{2000}   & \textbf{4000}   & \textbf{6000}   & \textbf{8000}   & \textbf{10000}  \\ \hline
\textbf{Average Overhead Ratio}        & 0.898 & 0.916 & 0.903 & 0.910 & 0.921 & 0.917 \\ \hline
\end{tabular}
}
\caption{Average Overhead Ratio for Joint Pruning.}
\label{tab:overhead_ratio}
\end{table}



\section{AP-Pulse$+$: Active Path Search}
Given a sub-problem instance $I=(In, Ex)$ and a set $\mathcal{T}$ of conflict Srlg sets, we use AP-Pulse$+$ to search for the min-cost active path $P_a$ that satisfies the following constraints (see line 7 in Algorithm \ref{algorithm:cose_pulse}):
\begin{enumerate}
    \item No Srlg in $I.Ex$ is included in $\Omega(P_a)$: $I.Ex\cap\Omega(P_a)=\emptyset;$
    \item Delay constraint: $d(P_a)\leq U$;
    \item All the Srlgs in $I.In$ must be in $\Omega(P_a)$: $I.In\subseteq\Omega(P_a);$
    \item For every conflict Srlg set $T\in \mathcal{T}$, $T\subsetneq\Omega(P_a)$.
\end{enumerate}

To obtain AP-Pulse$+$, we modify Pulse$+$ as follows:
\begin{enumerate}
    \item \textbf{Preparation stage}: For every Srlg $r\in I.Ex$, disable all the links contained in the Srlg $r$. This step ensures that the constraint (1) is met, i.e., $I.Ex\cap\Omega(P_a)=\emptyset$.
    \item \textbf{Path validation} (the box in line 7 of Algorithm \ref{algorithm:pulse_plus}):
    \begin{equation}\label{eqn:path_validation}
    \boxed{
    \begin{aligned}
    &d(P_{s\rightarrow t})\leq U
    \hspace{1mm}\text{and}\hspace{1mm} I.In\subseteq\Omega(P_{s\rightarrow t})
    \hspace{1mm}\text{and}\\
    & T\subsetneq\Omega(P_{s\rightarrow t})\text{ for any }T\in \mathcal{T}.
    \end{aligned}
    }
    \end{equation}
    This step ensures that the constraints (2)-(4) are met.
    \item \textbf{Prune strategy} (the box in line 15 of Algorithm \ref{algorithm:pulse_plus}): 
    \begin{equation}\label{eqn:ap_prune_strategy}
    \boxed{\begin{aligned}
    &d(P_{s\rightarrow u})+d(P_{u\rightarrow t}^{\text{min\_delay}})> U\\
    \text{or}\hspace{2mm}&c(P_{s\rightarrow u})+c(P_{u\rightarrow t}^{\text{min\_cost}})\geq \text{tmp\_min}&\\
    \text{or}\hspace{2mm} &\exists T\in \mathcal{T}\text{ such that }T\subseteq\Omega(P_{s\rightarrow u}).&\\
    \end{aligned}}
    \end{equation}
    Compared to the original pruning strategy (\ref{eqn:prune_strategy}), the above pruning strategy introduces the ``conflict check'': If a branch $P_{s\rightarrow u}$ contains a conflict Srlg set, then this branch is skipped. 
    
\end{enumerate}

%% file: sections/appendix3.tex
\section{Proofs}
\subsection{Proof of Theorem \ref{thm:optimality_pulse}}\label{appendix:proof_pulse}
\begin{proof}
We prove by contradiction. Suppose that the solution $P_{s\rightarrow t}^{\text{opt}}$ returned by Pulse$+$ is not optimal. Then, there must exist another path $P$ satisfying $L\leq d(P)\leq U$, such that $c(P)<c(P_{s\rightarrow t}^{\text{opt}})$. Consider the searching branch along the path $P$. At every intermediate node $u$ of the path $P$, we must have $d(P_{s\rightarrow u})+d(P_{u\rightarrow t}^{\text{min\_delay}})\leq d(P)\leq U$ and $c(P_{s\rightarrow u})+c(P_{u\rightarrow t}^{\text{min\_cost}})\leq c(P)<c(P_{s\rightarrow t}^{\text{opt}})\leq \text{tmp\_min}$. Hence, it is not possible to prune the path $P$'s branch based on the strategies in (\ref{eqn:prune_strategy}). Hence, Pulse$+$ should be able to find a solution with cost no more than $c(P)$. This leads to a contradiction. 
\end{proof}

\subsection{Proof of Theorem \ref{thm:conflict_set}}\label{appendix:proof_conflict}
\begin{proof}
Since every Srlg $r\in T$ is chosen within the set $\Omega(P_a)$, we must have $T\subseteq \Omega(P_a)$. We next show that every path $P$ satisfying $T\subseteq \Omega(P)$ does not have an Srlg-disjoint backup path that meets the delay-range requirement, i.e., $T$ is a conflict set. 

We prove by contradiction. Suppose that $P_a^{\prime}$ is a path satisfying $T\subseteq \Omega(P_a^{\prime})$ and $P_b^{\prime}$ is an Srlg-disjoint backup path of $P_a^{\prime}$ that meets the delay requirement. Clearly, $T\cap \Omega(P_b^{\prime})=\emptyset$. 

Consider the searching branch that yields the path $P_b^{\prime}$ in Algorithm \ref{algorithm:conflict_pulse_plus}. This branch must be able to reach its final stage (lines 10-15). According to Algorithm \ref{algorithm:conflict_pulse_plus}, a branch can be only cut in two places: lines 5-7 and lines 19-21. First, $P_b^{\prime}$ does not contain any link $e$ such that $e$ belongs to an Srlg in $T$. Hence, the searching branch of $P_b^{\prime}$ cannot be cut at lines 5-7. Second, $d(P_b^{\prime})\leq U$. Thus, the branch of $P_b^{\prime}$ cannot be cut at lines 19-21, either. 

When the $P_b^{\prime}$ branch reaches the final stage, it will not enter line 11; otherwise Algorithm \ref{algorithm:conflict_pulse_plus} will fail to return a conflict set. Then, at line 13, an Srlg $r\in \Omega(P_b^{\prime})\cap \Omega(P_a)$ will be chosen and added to $T$. This contradicts to the fact that $T\cap \Omega(P_b^{\prime})=\emptyset$. Based on the above discussion, $T$ must be a conflict set.
\end{proof}

\subsection{Proof of Theorem \ref{thm:optimality_cose_pulse}}\label{appendix:proof_cose_pulse}
\begin{proof}
We prove by contradiction. Suppose that the solution $(P_{a}^{\text{opt}}, P_{b}^{\text{opt}})$ returned by CoSE-Pulse$+$ is not optimal. Then, there must exist another pair of Srlg-disjoint path $(P_a, P_b)$ satisfying $d(P_a)\leq U$ and $d(P_a)-\delta\leq d(P_b)\leq \min\{U, d(P_a) + \delta\}$, such that $c(P_a)<c(P_{a}^{\text{opt}})$. Since $P_a$ has an Srlg-disjoint backup path, for every conflict set $T$ found in CoSE-Pulse$+$, we must have $T\subsetneq\Omega(P_a)$.

We have assumed that CoSE-Pulse$+$ terminates with a solution. Then, the total number of problem instances explored by CoSE-Pulse$+$ must be finite. We denote the set of explored problem instances by $\mathcal{I}$. Let $\mathcal{I}(P_a)\subseteq\mathcal{I}$ be the set of $I$'s such that $I.In\subseteq \Omega(P_a), I.Ex\cap \Omega(P_a)=\emptyset$. Clearly, $\mathcal{I}(P_a)$ is not empty because $(\emptyset,\emptyset)\in \mathcal{I}(P_a)$. Within $\mathcal{I}(P_a)$, there must be an $I^{\prime}\in \mathcal{I}(P_a)$ such that for every $I\in \mathcal{I}(P_a)$ and $I\neq I^{\prime}$, $I^{\prime}.In$ is not contained in $I.In$. Consider the problem instance $I^{\prime}$. Let $P_a^{\prime}$ be the AP-Pulse$+$ solution of $I^{\prime}$. Since $P_a$ satisfies all the requirements of $I^{\prime}$, we must have $c(P_a^{\prime})\leq c(P_a)<c(P_{a}^{\text{opt}})$. In addition, $P_a^{\prime}$ does not have an Srlg-disjoint backup path; otherwise, $P_{a}^{\text{opt}}$ would not be the optimal solution. Consider line 20 and line 22 of the CoSE-Pulse$+$ algorithm. Since $T\subsetneq \Omega(P_a)$ and $P_a\neq P_a^{\prime}$, there must exist a $k\in\{1,...,N\}$ such that $\{r_1,...,r_{k-1}\}\subseteq \Omega(P_a)$ and $r_k\notin \Omega(P_a)$. According to lines 24-26 of the CoSE-Pulse$+$ algorithm, a new problem instance $I_k^{\prime}=(I^{\prime}.In\cup \{r_1,...,r_{k-1}\}, I^{\prime}.Ex\cup \{r_{k}\})$ will be generated. It is easy to verify that $I_k^{\prime}\in \mathcal{I}(P_a)$ and $I^{\prime}.In\subseteq I_k^{\prime}.In$. This contradicts to the choice of $I^{\prime}$.
\end{proof}